\documentclass[12pt]{article}
\usepackage{amsfonts,amsmath,amsthm,amssymb,mathrsfs}
\usepackage{pb-diagram,pb-xy}
\usepackage[all]{xy}
\usepackage{algorithm}
\usepackage{algpseudocode}
\usepackage{listings}
\usepackage{color}
\usepackage{url}
\usepackage{float}
\usepackage{hyperref}
\newtheorem{theorem}{Theorem}
\definecolor{dkgreen}{rgb}{0,0.6,0}
\definecolor{gray}{rgb}{0.5,0.5,0.5}
\definecolor{mauve}{rgb}{0.58,0,0.82}
\usepackage{tensor}
\newcommand{\doi}[1]{\href{https://doi.org/#1}{doi: #1}}

\usepackage{framed}
\usepackage[breakable,skins,theorems]{tcolorbox}
\tcbset{
  exstyle/.style={},
}
\newtcbtheorem{myexample}{Example}{exstyle}{myex}

\algnewcommand{\Inputs}[1]{%
  \State \textbf{Inputs:}
  \Statex \hspace*{\algorithmicindent}\parbox[t]{.8\linewidth}{\raggedright #1}
}

\algnewcommand{\Outputs}[1]{%
  \State \textbf{Outputs:}
  \Statex \hspace*{\algorithmicindent}\parbox[t]{.8\linewidth}{\raggedright #1}
}
\newcommand{\StateLong}[1]{%
  \State \parbox[t]{\dimexpr\linewidth-\algorithmicindent}{%
    \hangindent=1em\hangafter=1 #1\strut}}

\algnewcommand\algorithmicinputb{\textbf{Precomputation:}}
\algnewcommand\Precomputation{\item[\algorithmicinputb]}

\algnewcommand\algorithmicinputc{\textbf{On-line computation:}}
\algnewcommand\Onlinecomputation{\item[\algorithmicinputc]}

\usepackage{import}
\usepackage{changepage}
\usepackage{xcolor}

\usepackage[shortlabels]{enumitem}
\usepackage{fullpage}
\usepackage{mathtools}
\usepackage{tikz}
\graphicspath{{figures/}}
\usetikzlibrary{calc,trees,positioning,arrows,chains,shapes.geometric,decorations.pathreplacing,decorations.pathmorphing,shapes,matrix,shapes.symbols,arrows.meta}
\usepackage{booktabs}

\newlist{steps}{enumerate}{1}
\setlist[steps, 1]{label = Step \arabic*:}

\def\C{{\mathbb C}}

\def\N{{\mathbb N}}

\def\R{{\mathbb R}}

\let\b=\boldsymbol
\def\bhat#1{\hat{\boldsymbol{#1}}}
\def\risingfac#1#2{\left(#1\right)^{\overline{#2}}}

\theoremstyle{definition}

\theoremstyle{plain}
\newtheorem{lemma}{Lemma}

\theoremstyle{definition}
\newtheorem{assumption}{Assumption}

\title{Fast Evaluation of Derivatives of Green's Functions Using Recurrences}

\usepackage{indentfirst}
\begin{document}
\tikzstyle{block} = [rectangle, draw, text width=10em, text centered, rounded      corners, minimum height=3em]
\author{Hirish Chandrasekaran, Andreas Kloeckner}
\maketitle 

\begin{abstract}
High-order derivatives of Green's functions are a key ingredient in Taylor-based fast multipole methods, Barnes-Hut $n$-body algorithms, and quadrature by expansion (QBX). In these settings, derivatives underpin the formation, evaluation, and/or translation of Taylor expansions. 

In this article, we provide hybrid symbolic-numerical procedures that generate recurrences to attain an $O(n)$ cost for the computation of $n$ derivatives (i.e. $O(1)$ per derivative) for arbitrary radially symmetric Green's functions. These procedures are general---only requiring knowledge of the PDE that the Green's function obeys. We show that the algorithm has controlled, theoretically-understood error.

We apply these methods to the method of quadrature by expansion, a method for the evaluation of layer potentials, which requires higher-order derivatives of Green's functions. In doing so, we contribute a new rotation-based method for target-specific QBX evaluation in the Cartesian setting that attains dramatically lower cost than existing symbolic approaches. 

Numerical experiments throughout support our claims of accuracy and cost.

\end{abstract}

\section{Introduction}
In this contribution, we develop an automated, symbolic procedure to evaluate high-order derivatives of Green's functions. The procedure takes as input a linear partial differential equation stated in Cartesian coordinates with polynomial coefficients in symbolic form. The Green's function $G$ is not directly an input to the procedure, but it is expected to have radial symmetry (i.e., $G$ depends on $\b{x}$ only through $|\b{x}|$) and to obey the PDE. Based on this information, we provide:
\begin{itemize}
\item an algorithm to compute a recurrence formula, also in symbolic form, that can be used to compute high-order derivatives of $G$, given a number of low-order derivatives,
\item a further algorithm to compute a recurrence formula for Taylor coefficients of the derivatives themselves, to be used in regions where the first recurrence loses accuracy due to rounding, and
\item theoretically-based error analysis that enables a hybrid scheme with controlled error.
\end{itemize}
More general radially symmetric functions are permitted as long as they satisfy the assumptions; however, the requirements of radial symmetry and satisfaction of a PDE appeared, at least to us, narrow enough to warrant the use of the more recognizable term ``Green's function.''

Thus far, the numerical computation of these derivatives has been fraught with difficulties, including: (1) given the wide variety in Green's functions used in practice, generalization across Green's functions is challenging, (2) symbolic computation, even with state-of-the-art systems, yields expressions of suboptimal asymptotic complexity, (3) evaluation of derivative formulas tends to encounter numerical instability.

Green’s functions are foundational tools in the analysis and numerical solution of partial differential equations (PDEs), especially in boundary integral methods. By using the potential of a point source, Green’s functions encode essential information about the underlying physics. In many advanced numerical algorithms—including the fast multipole method (FMM) and quadrature by expansion (QBX)—the computation of derivatives of Green’s functions plays a critical role.

As an example, a multipole expansion based on Taylor expansion in Cartesian coordinates takes the form
\[G(\b{x} - \b{y}) \approx \sum _{| p | \leqslant k
     } \underbrace{\frac{D^p_{\b{y}} G
   (\b{x} - \b{y}) |_{\b{y} = \b{c}}
    }{p!}}_{\text{basis}}
   \underbrace{(\b{y} - \b{c})^p}_{\text{coefficient}},
  \]
where $p$ is to be viewed as a multi-index (see Section~\ref{sec:notation}{} for notation). If the overall order $k$ in this expansion is chosen to be sufficiently high, then many high-order derivatives of the Green's function are required for evaluation of the basis. The recurrence formulas in this contribution allow the computation of each additional derivative based on prior ones at a fixed cost.

Multipole expansions, meanwhile, are a key building block in the fast multipole methods. These accelerate pairwise interactions in $N$-body problems and boundary integral equations for a given kernel. In many cases, expansions used in an FMM implementation are only applicable for a single kernel, making them \emph{kernel-specific} (e.g., \cite{greengard1987fast, greengard1988rapid, greengard1988efficient, greengard2002new}). A number of approaches exist that sidestep this, yielding \emph{kernel-independent} FMMs
\cite{ying2004kernel, zhao1987n, oppelstrup2013matrix, shanker2007accelerated, coles2020optimizing, fernando2023automatic}. Of the kernel-independent FMMs formulated, \cite{zhao1987n, oppelstrup2013matrix, shanker2007accelerated, coles2020optimizing, fernando2023automatic} are based on Taylor series, and these represent the subfamily of approaches to which the methods introduced here are applicable.

\subsection{Green's Functions}
\label{sec:green}
Let $\mathcal{L}$ be an order $c \in \N$ PDE operator such that
\begin{equation}
\label{eq:pdeweak}
\mathcal{L} G(|\b{x}|) = -\delta(\b{x})
\end{equation}
where $\b{x} = (x_1, \dots ,x_d) \in \R^d$ and we interpret \eqref{eq:pdeweak} in a weak sense (cf.~\cite{evans2010partial}, Section 2.2). In the case that $\bhat v = (v_1, \dots, v_d)$ is a vector in $\R^d$, its norm is given by the $l_2$ norm: $|(v_1, \dots, v_d)|_2 = \sqrt{v_1^2 + \dots + v_d^2}$.

We assume $\mathcal{L}$ and $G(|\b{x}|)$ have rotational symmetry. For $G(|\b{x}|)$, rotational symmetry is embedded in its representation since it is the composition of a scalar function $G$ and $|\b{x}|$, where $|\b{x}|$ is rotationally symmetric.
Then if $f \in C^c(\R^d)$, a solution of
\[
\mathcal{L} u  = -f
\]
is given by $u = G * f$, where we let $*$ represent the convolution operator, since:
\[
\mathcal{L} (G * f)= (f * \mathcal{L} G) = f * (-\delta) = -f,
\]
We consider PDEs that have polynomial coefficients:
\[
\mathcal{L}= \sum_{\b{q} \in \mathcal{M}(c)} p_q(\b{x})  \partial^{\b{q}}_{\b{x}}
\]
where $p_q(\b{x}) \in \C[x_1, \dots, x_d]$.
If the PDE has transcendental coefficients, then a local polynomial approximation can be constructed at the evaluation point of interest, however methods for doing so are beyond the scope of the present work.

\begin{table}[htbp]
   \centering
   \begin{tabular}{@{} lc @{}} 
      \toprule
      PDE    & Radially Symmetric Green's Function \\
      \midrule
      Laplace 2D      & $-\frac{\log(|{\b x}|)}{2\pi}$ \\
      Laplace 3D       & $-\frac{1}{4\pi|{\b x}|}$  \\
      Helmholtz 2D      &  $\frac{i}{4} H_0^{(1)}(k |{\b x}|)$     \\
      Helmholtz 3D       & $\frac{e^{i k |{\b x}|}}{4\pi |{\b x}|}$  \\
      Yukawa 2D       & $\frac{1}{2\pi} K_0(k|{\b x}|) $  \\
      Yukawa 3D       & $\frac{e^{-k|{\b x}|}}{4\pi |{\b x}|} $  \\
      Biharmonic 2D       & $\frac{|{\b x}|^2}{8\pi} \log| {\b x}|$  \\
      Biharmonic 3D       & $-\frac{|{\b x}|}{8 \pi} $  \\
      \bottomrule
   \end{tabular}
   \caption{A \textit{non}-exhaustive list of applicable radially symmetric Green's functions.}
   \label{tab:booktabs}
\end{table}

\subsection{Derivatives of a Green's Function}
Let $G(|{\b x}|)$ be a radially symmetric Green's function where ${\b x} =(x_1, \dots, x_d) \in \R^d$. The goal of this paper is to present an algorithm to find recurrences for the $n$th directional derivative of $G(|\b{x}|)$ in the $\hat{x}_1$ direction:
\begin{equation}
\label{eq:we-want-recurrence}
\frac{d^n}{dt^n}|_{t=0} G(|{\b x} + \hat{x}_1 t|) = \partial^n_{x_1} G(|{\b x}|) \qquad  (n \in \N),
\end{equation}
where $\hat{x}_1$ is the elementary unit vector along the $x_1$-axis. Observe that, due to the rotational symmetry of $G$, it is sufficient to only consider derivatives in the $\hat{x}_1$-direction, as any derivative direction can be rotated to align with $\hat{x}_1$, as shown in Figure~\ref{fig:rotational-sym}.

\begin{figure}[!ht]

\centering
\begin{tikzpicture}[scale=0.75]
\draw [->, thick] (0,0) -- (5,0) node [anchor=north,pos=1] {${x}_1$};
\foreach\rad in {1,2,3,4}
  \draw  [ dashed](0,0) circle (\rad);
\foreach\i/\rot/\anchor in {1/0/south,2/40/south,3/80/east}
	\draw [->,rotate=\rot] (1,1) -- ++(2,0) node [pos=0.5,anchor=\anchor] {$\hat v_\i$};
\foreach\i/\rot/\anchor in {1/0/south,2/40/south,3/80/east}
\filldraw[black,anchor=\anchor,rotate=\rot] -- ++(2,1) circle (2pt);
\end{tikzpicture}

\caption{Directional derivatives taken at three points along $\hat{v}_1, \hat{v}_2, \hat{v}_3$ are identical due to radial symmetry of $G(|\b{x}|)$.}
\label{fig:rotational-sym}
\end{figure}

Suppose $r \in \R_{ > 0}$. The derivatives
\[
\frac{d^n}{dr^n} G(r) \qquad  (n \in \N),
\]
are a subset of the derivatives considered in \eqref{eq:we-want-recurrence} (simply confine ${\b x}$ to the $x_1$-axis). As an example, one may use our methods to obtain a recurrence for the derivatives of the Hankel function of the first kind of order 0 (see Table \ref{tab:booktabs}).

\subsection{Prior Work and Novelty}
In this work, we make the following contributions:
\begin{itemize}
\item For arbitrary Green's functions that satisfy a PDE with polynomial coefficients (as described in Section~\ref{sec:green}),
provide a method that computes $n$ derivatives of type \eqref{eq:we-want-recurrence} at a cost of $O(n)$ floating point operations, i.e., $O(1)$ per derivative.
This is achieved through the symbolic computation of a recurrence relation for the derivatives.
\item We identify regions in which the evaluation of the recurrence encounters numerical issues mainly due to cancellation. To address the instability, we propose a novel alternate evaluation scheme making use of Taylor expansions of derivatives about an axis.
\item We provide a model of the error incurred in both types of evaluation, motivating a hybrid scheme. Our error model leads to an error bound that unconditionally controls the error incurred in the hybrid scheme.
\item We provide numerical experiments that numerically validate our method development and theoretical findings.
\item We demonstrate the integration of our method into QBX. This integration alone provides an asymptotic cost reduction, that is particularly noticeable, for example, in the case of the Helmholtz PDE. The integration is algorithmically straightforward and has a negligible impact on overall method error.
\item Compounding with the previous point, we provide a new rotation-based method for target-specific QBX evaluation in the Cartesian setting that attains dramatically lower cost than existing symbolic approaches, for all kernels.
\end{itemize}

In the following, we discuss connections between the contributions above and various existing approaches in the literature.

A mixed derivative is an arbitrary mixed partial derivative expression, i.e.
\[
\partial^{n_1}_{x_1} \partial^{n_2}_{x_2} \dots \partial^{n_d}_{x_d} G(|\b{x}|)
\] where $n_1, \dots, n_d \in \N_{0}$. Tausch \cite{tausch2003fast} proposes a method to calculate mixed derivatives of radially symmetric kernels given knowledge of the set of radial derivatives beforehand
\begin{equation}
\label{eq:tausch}
\left\{\left( \frac{1}{r} \frac{d}{dr} \right)^i G(r) \right\}_{i \in\{0, \dots, p\}}
\end{equation}
using a recurrence formula as well. Suppose $p \in \N$ and let $\mathcal{D}^p$ represent all mixed derivatives of order less than or equal to $p$. There are $O(p^{d})$ derivatives in $\mathcal{D}^p$. Their algorithm provides a way to compute the $O(p^d)$ derivatives in $\mathcal{D}^p$ with $O(p^{d+1})$ work. This means the cost per derivative is about $O(p)$. Furthermore, in order to get this asymptotic cost, Tausch considers the specific case of the Laplace equation where a recurrence is utilized to compute the radial derivatives in \eqref{eq:tausch}. Given an arbitrary kernel, it is not clear a priori if a recurrence for the radial derivatives exists.

In \cite{slevinsky2011new}, higher order derivative formulas are derived that can be applied to $k$-times differentiable functions such that
\[
\left( \frac{d}{x^m dx} \right)^k (x^{-n} G(x)),
\]
using the notation of \cite{slevinsky2011new} for the operator
\[\frac{d}{x^m dx} := \frac{1}{x^m}\frac{d}{dx},\]
are well-defined for some $m, n \in \N$. Their work gives a recurrence formula that is $O(n)$ per derivative. This formula follows from the product rule: a function obeying an ODE has an associated recurrence formula for its derivatives. To see why, if $f(x)$ is a function satisfying a differential equation of the form
\[
f^{(m)}(x) = \sum_{k=0}^{m-1} p_k(x) f^{(k)}(x),
\]
then the derivatives of $f(x)$ can be represented as:
\begin{equation}
\label{eq:diff-eq-formula}
f^{(m+n)}(x) = \sum_{k=0}^{m-1} \sum_{i=0}^n \binom{n}{i} p_k^{(i)}(x) f^{(k+n-i)}(x).
\end{equation}
We make use of \eqref{eq:diff-eq-formula} to a similar end.

Kauers et al. \cite{kauers2011concrete} establish a relationship between holonomic recurrences and holonomic power series. A \textit{holonomic} recurrence of order $r$ and degree $d$ is one such that there exist polynomials $p_0(x), \dots, p_r(x) \in \mathbb{K}[x]$ of degree at most $d$ with $p_0(x) \ne 0 \ne p_r(x)$ such that
\[
p_0(n) a_n + p_1(n) a_{n+1} + \dots + p_r(n) a_{n+r} = 0
\]
and a power series $a(x) \in \mathbb{K}[[x]]$ is called \textit{holonomic} of order $r$ and degree $d$ if there exist polynomials $q_0(x), \dots, q_r(x) \in \mathbb{K}[x]$ of degree at most $d$ with $q_0(x) \ne 0 \ne q_r(x)$ such that
\begin{equation}
\label{eq:holonomic}
q_0(x) a(x) + q_1(x) \partial^x a(x) + \dots + q_r(x) \partial^r_x a(x) = 0
\end{equation}
where \eqref{eq:holonomic} describes a \textit{holonomic differential equation} of order $r$ and degree $d$.
Every holonomic power series has coefficients that satisfy a holonomic recurrence. This follows from inserting the power series representation of $a(x) = \sum_i a_i x^i$ into \eqref{eq:holonomic}, and then collecting coefficients. The resulting expression is a holonomic recurrence for $a_i$. Note that the coefficients of a power series (centered at $x=0$) are clearly related to its derivatives at $x=0$, which means this gives a relationship between derivatives of a function and the associated holonomic differential equation it satisfies. This conveys the same information as \eqref{eq:diff-eq-formula}.

In \cite{zhang2013fast}, a differential algebra framework is used to compute derivatives of algebraic functions. A differential algebra extends the definitions of a ring to include a derivation, an operator that models a derivative. A differential algebra framework 
allows one to evaluate all $n$ derivatives of a multivariate scalar function of $v$ variables, say $f(x_1, \dots, x_v)$, as a finite set of arithmetic operations in the differential algebra $\tensor*[^{}_n]{{D}}{^{}_v}$. The complexity of this method is asymptotic with the complexity to compute elementary operations ($+, \times, -, \div$) in $\tensor*[^{}_n]{{D}}{^{}_v}$. The number of floating point operations to compute the multiplication of two elements in $\tensor*[^{}_n]{{D}}{^{}_v}$  \cite{hawkes1999modern} is
\[
\frac{(n+2v)!}{n!(2v)!}.
\]
This means the cost to compute $O(n^v)$ derivatives is $O(n^{2v})$, and the cost per derivative amortized is $O(n^v)$.

In \cite{fernando2023automatic}, the existence of mixed derivative recurrences for the 2D/3D Laplace and Helmholtz case was shown. These recurrences were generated in a guess-and-check manner and lower the cost per derivative to $O(1)$. \cite{fernando2023automatic} uses those recurrences as a crucial ingredient in lowering the operation count of certain translation operators.
Our work can be seen as extending the reach of these methods by providing a means for automatically identifying
recurrences for a broad class of operators.

\subsection{Overview of the Paper}

In Section~\ref{sec:making-recurrences}, we first define some fundamental notation and then discuss the derivation of a recurrence for a radially-symmetric function obeying a PDE. A consideration of empirical error behavior of this recurrence motivates the introduction of an alternate computation approach for derivatives as well as a detailed study of the error behavior of both approaches, leading to a hybrid method with controlled error throughout. In Section~\ref{sec:qbx}, we discuss the relevance of our recurrences to the QBX method, which we first review. We provide an algorithmic approach analogous to `target-specific QBX' that can drastically reduce computational cost in QBX, compounding with savings from the use of the recurrence. In Section~\ref{sec:conclusions}, we close with some remarks on potential impacts as well as directions for future work.

\section{Methodology to Produce Recurrences}
\label{sec:making-recurrences}

\subsection{Prerequisite Notation and Concepts}
\label{sec:notation}
We review a number of fundamental notions that will be useful throughout the developments of the paper, including multi-index notation and partitions/vector partitions as well as a number of Faa Di Bruno formulas. Let ${\b x} = (x_1, \dots, x_d) \in \R^d$ and let there be a multi-index $\b\alpha = (\alpha_1, \dots, \alpha_d) \in \N^d_{\b{0}}$ and $\b\beta = (\beta_1, \dots, \beta_d) \in \N^d_{\b{0}}$. Then
\begin{align}
\partial^{\b\alpha}_{{\b x}} &:= \partial^{\alpha_1}_{x_1} \dots \partial^{\alpha_d}_{x_d},\\
\b\alpha \leq \b\beta &:\Leftrightarrow \alpha_i \leq \beta_i \qquad  \qquad(i \in \{1, \dots, d\}),\\
\b\alpha ! &:= \alpha_1 ! \dots \alpha_d !,\\
{\b x}^{\b\alpha} &:= x_1^{\alpha_1} \dots x_d^{\alpha_d},\\
|\b{\alpha}| &:= |\b{\alpha}|_1 =  \sum_i \alpha_i. \label{eq:mi-magnitude}
\end{align}
As a way of fixing notation for the magnitude $|\cdot|$, \eqref{eq:mi-magnitude} implies that when $\b{\alpha} = (n_1, \dots, n_d)$ is a multi-index in $\N^d_{\b{0}}$, its magnitude is given by the $l^1$ norm: $|(n_1, n_2, \dots, n_d)|_1 = n_1 + n_2 + \dots + n_d$. By contrast, when $\bhat v = (v_1, \dots, v_d)$ is a vector in $\R^d$, its norm is given by the $l^2$ norm: $|(v_1, \dots, v_d)|_2 = \sqrt{v_1^2 + \dots + v_d^2}$.

Let the set of multi-indices with magnitude less than or equal to $c \in \N$ be denoted by
\[
\mathcal{M}(c) := \left\{ \b{q} = (q_1, \dots, q_d) \in \N^d_{\b{0}} : \sum_{i=1}^d q_i \leq c\right\}.
\]
Let $\b{0}$ be the zero vector in $d$ dimensions. A \emph{vector partition} of $\b{\alpha} \in \N^d_{\b{0}}$ is a set of vectors in $\N^d_{\b{0}} \setminus \b{0}$ that add up to $\b{\alpha}$. It can be represented by a function $\pi : \N^d_{\b{0}} \to \N_{0}$, such that if $\b{\beta} \in \N^d_{\b{0}}$, $\pi(\b{\beta})$ represents the number of times $\b{\beta}$ appears in the vector partition. Clearly $\pi$ maps to 0 except on a finite set representing the unique elements of the partition.  Let $P_{\b{\alpha}}$ denote the set of vector partitions of $\b{\alpha}$. Let $\pi \in P_{\b{\alpha}}$. We define the following operations on vector partitions:
\begin{align}
\label{eq:cardinality-part}
|\pi| &= \sum_{\b{\beta} \in \N^{d}_{\b{0}}} \pi(\b{\beta}) \qquad  (\text{cardinality}),\\
\pi! &= \prod_{\b{\beta} \in \N^{d}_{\b{0}}} \pi(\b{\beta})!.
\end{align}
If $k \in \N$, we denote as $P_{\b{\alpha}, k}$ the set of partitions of $\b{\alpha}$ with cardinality $k$:
\[
P_{\b{\alpha}, k} = \{ \pi \in P_{\b{\alpha}} : |\pi| = k\}.
\]
As a mild generalization, we also define the notation $P_{l,m}$ with $l\in\mathbb N$ consistent with the above, where $l$ should be viewed as a multi-index of length 1. From \cite{turcu2020vector}, we have the following Faa Di Bruno rule for $n$ derivatives of a composition of functions $f : \R \to \R$ and $g : \R^d \to \R$ where ${\b z} \in \R^d$:
\begin{equation}
\label{eq:multi-faa-di-bruno-2}
\partial^{\b\alpha}_{\b z} f(g({\b z})) = \sum_{k=1}^{|\alpha|} f^{(k)}(g(\b{z})) \sum_{\pi \in P_{\b{\alpha},k}} \frac{\b{\alpha}!}{\pi!} \prod_{\b{\beta} \leq \b{\alpha}, |\b{\beta}| \leq |\b{\alpha}|-k+1} \left( \frac{\partial^{\b{\beta}}_{\b{z}} g(\b{z})}{\b{\beta}!} \right)^{\pi(\b{\beta})}.
\end{equation}
We further require from \cite{functions.wolfram.com} a rule for $n$ derivatives of a composition with a squared input:
\begin{equation}
\label{eq:chain-rule-square}
\frac{\partial^n f(z^2)}{\partial z^n} = \sum_{k=0}^n \frac{\risingfac{2k-n+1}{2(n-k)}}{(n-k)!(2z)^{n-2k}} f^{(k)}(z^2)
\end{equation}
and again from \cite{functions.wolfram.com} a rule for $n$ derivative of a composition with a square root input:
\begin{equation}
\label{eq:chain-rule-root}
\frac{\partial^n f(\sqrt{z})}{\partial z^n} = \sum_{k=0}^n \frac{(-1)^{n-k} \risingfac{k}{2(n-k)}}{(n-k)!(2 \sqrt{z})^{2n-k}} f^{(k)}(\sqrt{z})
\end{equation}
where
\[\risingfac a n := \prod_{k=0}^{n-1} (a+k)\]
is the rising factorial or Pochhammer symbol.

Another expression from \cite{functions.wolfram.com} that is a result of the product rule is
\begin{equation}
\label{eq:productrule1}
\frac{d^n}{dx^n}  x^p f(x) = \sum_{l=0}^p \left(\prod^{l-1}_{k=0}(n-k)\right) \binom{p}{l} x^{p-l} \frac{d^{n-l}}{dx^{n-l}} f(x)\qquad (n \geq 0)
\end{equation}
 if $n \in \N_{0}$, $p \in \N$ and $f: \R \to \R$, $f \in C^n(\R)$.
\subsection{PDE to ODE}
\label{sec:pde-to-ode}
Given a PDE with polynomial coefficients satisfied by $G(|\b{x}|)$, we derive an ODE satisfied by $G(|\b{x}|)$ with respect to $x_1$, treating $(x_2, \dots, x_d)$ as parameters. The ODE coefficients depend on $(x_2, \dots, x_d)$, so the equation holds parametrically: for each fixed choice of $(x_2, \dots, x_d)$, it is an ODE in $x_1$. The main result of this section (Theorem \ref{thm:pdeetoode}) is that if the PDE has coefficients in $\C[x_1, \dots, x_d]$, then this ODE has coefficients in $\C[x_1, \dots, x_d]$. 

To derive the ODE, we take the following steps:
\begin{itemize}
\item Substitute all spatial derivatives of the PDE with radial derivatives applying the multidimensional Faa Di Bruno chain rule \eqref{eq:multi-faa-di-bruno-2}.
\item Take these radial derivatives and substitute them for partial $x_1$ derivatives using the Faa Di Bruno formula when $d=1$.
\item The end result is an ODE in $x_1$ satisfied by $G(|\b{x}|)$ with polynomial coefficients $x_1, \dots, x_d$. 
\end{itemize}

To illustrate the process, we consider the case of the Laplace PDE in two dimensions as a simple example.

\begin{myexample}{Laplace in 2D}{laplace-ode}
Starting with the Laplace PDE in 2D we derive an ODE in $x_1$ satisfied by $G(|\b{x}|)$. Suppose $\b{x} = (x_1, x_2) \in \R^2$, $r = |\b{x}|$, and $G(|\b{x}|)$ is the Green's function for the Laplace PDE in 2D:
\[
\partial^2_{x_1} G(|\b{x}|) + \partial^2_{x_2} G(|\b{x}|) = 0, \qquad \b{x} \ne 0.
\]
\begin{enumerate}
 \item Substitute $\partial_{x_i} = \left(\partial_{x_i} r\right) \partial_r$ for $i \in \{1, 2\}$
\[
 \begin{aligned}
[\left(\partial_{x_1} r\right) \partial_r]^2 G(|\b{x}|) + [\left(\partial_{x_2} r\right) \partial_r]^2 G(|\b{x}|) = 0, \qquad \b{x} \ne 0 \\
\frac{1}{r} \partial_rG(r) + \partial^2_rG(r) = 0, \qquad r \ne  0 .
\end{aligned}
\]
 \item Substitute $\partial_r = \left(\partial_r x_1\right) \partial_{x_1}$
\[
\frac{1}{r} [\left(\partial_r x_1\right) \partial_{x_1}]G(r) + [\left(\partial_r x_1\right) \partial_{x_1}]^2G(r) = 0, \qquad r \ne  0
\]
\[
\left(\frac{x_2^2}{x_1^2} + 1\right)\partial^2_{x_1} G(|\b{x}|) + \left(\frac{x_1^2 - x_2^2}{x_1^3} \right)\partial_{x_1} G(|\b{x}|) = 0, \qquad  \b{x} \ne 0.
\]
\item Normalize with multiplication by $x_1^3$:
\begin{equation}
\label{eq:normalized-laplace}
\left(x_2^2 x_1 + x_1^3\right)\partial^2_{x_1} G(|\b{x}|) + \left(x_1^2 - x_2^2\right)\partial_{x_1} G(|\b{x}|) = 0, \qquad  \b{x} \ne 0.
\end{equation}
\eqref{eq:normalized-laplace} is an ODE satisfied by the first and second $x_1$-derivatives of $G$.
\end{enumerate}
\end{myexample}

Suppose we have that $\mathcal{L}_{\b{x}}G(|{\b x}|) = 0$ for $\b{x} \in \R^d \setminus \{0\}$ where $\mathcal{L}_{\b{x}}$ is a linear PDE of order $c \in \N$ with multivariate polynomial coefficients $p_{\b{q}}(\b{x}) \in \C[x_1, \dots, x_d]$:
\begin{equation}
\label{eq:PDE}
\mathcal{L}_{\b{x}} G(|{\b x}|)=\sum_{\b{q} \in \mathcal{M}(c)} p_{\b{q}}({\b x}) \partial^{\b{q}}_{{\b x}} G(|{\b x}|)  = \delta(\b{x}).
\end{equation} 
Then Theorem \ref{thm:pdeetoode} below shows that \eqref{eq:PDE} together with \eqref{eq:normalization:factor} gives rise to an ODE in $x_1$ with coefficients in $\C[x_1, \dots, x_d]$ satisfied by $G(|\b{x}|)$. We begin with an intermediate technical result.

\begin{lemma}
\label{lem:dx1dr}
Suppose $l, m \in \N$, $m \leq l$. If $\pi \in P_{l,m}$, then
\[
\prod_{1 \leq j \leq l-m+1} \left( \frac{\partial^j_r x_1}{j!} \right)^{\pi(j)} = \frac{1}{r^l x_1^{2l-m}}  e_{\pi,l}(\b{x}),
\]
where $e_{\pi,l}(\b{x}) \in \C[x_1, \dots, x_d]$.
\end{lemma}
\begin{proof}
Let $x_2,\dots,x_d$ be fixed.
Let $h(\xi) = \sqrt{\xi-(x_2^2 + \dots + x_d^2)}$. Then, if $x_1 > 0$, $x_1$ is the composition of $h$ and a square of $r$, that is, $h(r^2) = |x_1| = x_1$. If $x_1 < 0$, we consider the alternate case $\tilde h(\xi) = -\sqrt{\xi-(x_2^2 + \dots + x_d^2)}$ which can be treated analogously. We will assume $x_1 > 0$ without loss of generality. Next, rewrite $\partial^j_{r} x_1 = \partial^j_{r} h(r^2)$ using \eqref{eq:chain-rule-square}:
\begin{equation}
\label{eq:partialjr}
\partial^j_{r} x_1 =\partial^j_{r} h(r^2)=\sum_{i=0}^j \frac{\risingfac{2i-j+1}{2(j-i)}}{(j-i)!(2r)^{j-2i}} h^{(i)}(r^2).
\end{equation}
A simple inductive proof shows that
\[
h^{(i)}(\xi) =c_i [\xi - (x_2^2 + \dots + x_d^2)]^{\frac{1}{2}-i},
\qquad \text{with}\qquad
c_i =\begin{cases} 1 &i=0\\
 \frac{(-1)^{i-1}(2i-3)!!}{2^i} & i \geq 1.
\end{cases}
\]
This implies that $h^{(i)}(r^2)$ is of the form $c_i x_1^{1-2i}$ where $c_i \in \R$. We can then rewrite \eqref{eq:partialjr} as
\begin{equation}
\label{eq:partialjrsimp}
\partial^j_{r} x_1  = \sum_{i=0}^j \frac{\risingfac{2i-j+1}{2(j-i)}}{(j-i)!(2r)^{j-2i}} c_i \frac{1}{x_1^{2i-1}} = \frac{1}{r^j x_1^{2j-1}}\sum_{i=0}^j \frac{r^{2i}\risingfac{2i-j+1}{2(j-i)}}{(j-i)! \, 2^{j-2i}} c_i x_1^{2j-2i}.
\end{equation}
Let
\[
w_{j}(\b{x}) = \sum_{i=0}^j \frac{r^{2i}\risingfac{2i-j+1}{2(j-i)}}{(j-i)! \, 2^{j-2i}} c_i x_1^{2j-2i}.
\]
$w_{j}(\b{x})$ is a polynomial in $x_1, \dots, x_d$ due to the non-negative powers of $x_1$ and even powers of $r$ present in its expression. We can then rewrite \eqref{eq:partialjrsimp} as
\begin{equation}
\label{eq:partialjfinalsimp}
\partial^j_{r} x_1  = \frac{1}{r^j x_1^{2j-1}} w_{j}(\b{x}).
\end{equation}
Suppose $\pi \in P_{l,m}$. Let $j_{\pi,1}, \dots, j_{\pi, m} \in \N$ correspond to the partition of $l$, $\pi$, with duplicates ($j_{\pi,1}, \dots, j_{\pi, m}$ are possibly non-unique). It follows that $j_{\pi,1} + \dots + j_{\pi,m} = l$. Then it must be the case using \eqref{eq:partialjfinalsimp} that:
\begin{align*}
 \prod_{1 \leq j \leq l-m+1} \left( \frac{\partial^j_r x_1}{j!} \right)^{\pi(j)}
&=  \prod_{1 \leq j \leq l-m+1} \left( \frac{1}{r^j x_1^{2j-1}} \frac{w_{j}(\b{x})}{j!} \right)^{\pi(j)}\\
&= \left( \frac{1}{r^{j_{\pi,1}} x_1^{2j_{\pi,1}-1}} \frac{w_{j_{\pi,1}}(\b{x})}{j_{\pi,1}!} \right) \cdots
 \left( \frac{1}{r^{j_{\pi,m}} x_1^{2j_{\pi,m}-1}} \frac{w_{j_{\pi,m}}(\b{x})}{j_{\pi,m}!} \right).
\end{align*}
Collecting powers of $r$ gives an exponent of $j_{\pi,1} + \dots + j_{\pi,m} = l$. Collecting powers of $x_1$ gives exponent $(2j_{\pi,1}-1) + \dots + (2j_{\pi,m}-1) = 2l - m$. Therefore,
\begin{equation}
\label{eq:partialjrprodsimp}
\prod_{1 \leq j \leq l-m+1} \left( \frac{\partial^j_r x_1}{j!} \right)^{\pi(j)} = \frac{1}{r^l x_1^{2l-m}}  e_{\pi,l}(\b{x}),
\end{equation}
where
\[e_{\pi,l}(\b{x}) = \prod_{p=1}^{m} \frac{w_{j_{\pi,p}}(\b{x})}{j_{\pi,p}!}\]
is a polynomial in $x_1, \dots, x_d$.
\end{proof}

As in the proofs above, throughout the remainder of the article, $x_2, \dots, x_d$
are assumed constant.

\begin{theorem}
\label{thm:pdeetoode}
Suppose $G : \R \to \C$, $\b{x} \in \R^d$, $\b q\in \N_0^d$, and $r = |\b{x}|$. Then there exist $n_1, n_2 \in \N_0$, $\lambda_i \in \C[x_1, \dots, x_d]$ depending on $\b q$ such that
\begin{equation}
\label{eq:normalization:factor}
r^{2n_1} x_1^{n_2} \partial^{\b{q}}_{\b{x}} G(|\b{x}|) = \sum_{i=0}^{|\b{q}|} \lambda_i(\b{x})\partial^i_{x_1} G(|\b{x}|).
\end{equation}
\end{theorem}

\begin{proof}
Let $x_2,\dots,x_d$ be fixed.
We can write $G(|\b{x}|)$ as the composition of $f(\xi) = G( \sqrt{\xi})$ for $\xi \in \R^{+}_0$ and $g(\b{z}) = z_1^2 + \dots z_d^2$ for $\b{z} = (z_1, \dots, z_d) \in \R^d$ so that $G(|\b{x}|) = (f \circ g)(\b{x})$. We can then apply \eqref{eq:multi-faa-di-bruno-2} to $\partial^{\b{q}}_{\b{x}} G(|\b{x}|) = \partial^{\b{q}}_{\b{x}} \left((f\circ g)(\b{x}) \right)$ and get
\begin{equation}
\label{eq:faa-di-bruno-applied-once}
\partial^{\b{q}}_{\b{x}} G(|\b{x}|) = \partial^{\b{q}}_{\b{x}} (f\circ g)(\b{x})  = \sum_{k=1}^{|\b{q}|} f^{(k)}(g(\b{x})) \sum_{\pi \in P_{\b{q},k}} \frac{\b{q}!}{\pi!} \prod_{\b{\beta} \leq \b{q}, |\b{\beta}| \leq |\b{q}|-k+1} \left( \frac{\partial^{\b{\beta}}_{\b{x}} g(\b{x})}{\b{\beta}!} \right)^{\pi(\b{\beta})}.
\end{equation}
Note that since $g(\b{x}) = x_1^2 + \dots + x_d^2$, $\partial^{\b{\beta}}_{\b{x}} g(\b{x})$ is polynomial in $x_1, \dots, x_d$.
Let
\begin{equation}
\label{eq:bqk-def}
b_{\b{q},k}(\b{x}) = \sum_{\pi \in P_{\b{q},k}} \frac{\b{q}!}{\pi!} \prod_{\b{\beta} \leq \b{q}, |\b{\beta}| \leq |\b{q}|-k+1} \left( \frac{\partial^{\b{\beta}}_{\b{x}} g(\b{x})}{\b{\beta}!} \right)^{\pi(\b{\beta})},
\end{equation}
which is a polynomial in $x_1, \dots, x_d$.
Then rewrite \eqref{eq:faa-di-bruno-applied-once} as
\begin{equation}
\label{eq:partial-q-x}
\partial^{\b{q}}_{\b{x}} G(|\b{x}|) = \sum_{k=1}^{|\b{q}|} f^{(k)}(g(\b{x})) b_{\b{q},k}(\b{x}).
\end{equation}
Next, we consider $f^{(k)}(g(\b{x}))$. If $\xi \in \R^{+}_0$, we can apply \eqref{eq:chain-rule-root} to $f^{(k)}(\xi)$ (recall that $f(\xi)= G(\sqrt{\xi})$) and obtain
\begin{equation}
\label{eq:applied-bruno-sqrt}
f^{(k)}(\xi) = \sum_{l=0}^k \frac{(-1)^{k-l}\risingfac{l}{2(k-l)}}{(k-l)!(2 \sqrt{\xi})^{2k-l}} G^{(l)}(\sqrt{\xi}).
\end{equation}
We define $r =\sqrt{ x_1^2 + \dots x_d^2}$ so that if we use the identity $g(\b{x}) = r^2$ and substitute \eqref{eq:applied-bruno-sqrt} into \eqref{eq:partial-q-x} along with $\xi = r^2$, we find
\begin{equation}
\label{eq:partialqxbulky}
\partial^{\b{q}}_{\b{x}} G(|\b{x}|) =\sum_{k=1}^{|\b{q}|} \sum_{l=0}^k \frac{(-1)^{k-l}\risingfac{l}{2(k-l)}}{(k-l)!(2 r)^{2k-l}} G^{(l)}(r) b_{\b{q},k}(\b{x}) .
\end{equation}
We have now translated all derivatives in $x_1, \dots, x_d$ into derivatives in $r$. Next, we will translate the derivatives in $r$ to derivatives in $x_1$. Define
\begin{align*}
  \tilde G &:  x_1 \mapsto G(|[x_1,x_2,\dots,x_d]^T|),\\
  \eta &: r \mapsto x_1(r).
\end{align*}
Then apply \eqref{eq:multi-faa-di-bruno-2} to the composition $\tilde G \circ \eta$:
\begin{align}
\label{eq:glr}
\partial^l_r G(r)  =
(\tilde G \circ \eta)^{(l)}(r)
&= \sum_{m=1}^l \tilde G^{(m)}(\eta(r)) \sum_{\pi \in P_{l,m}} \frac{l!}{\pi!} \prod_{1 \leq j \leq l-m+1} \left( \frac{\eta^{(j)}(r)}{j!} \right)^{\pi(j)}\\
&= \sum_{m=1}^l \partial^m_{x_1}G(|\b{x}|) \sum_{\pi \in P_{l,m}} \frac{l!}{\pi!} \prod_{1 \leq j \leq l-m+1} \left( \frac{\partial^j_r x_1}{j!} \right)^{\pi(j)}.\notag
\end{align}
Using the result of Lemma~\ref{lem:dx1dr} by substituting \eqref{eq:partialjrprodsimp} into \eqref{eq:glr} gives
\begin{align}
\partial^l_r G(r)
&= \sum_{m=1}^l \partial^m_{x_1}G(|\b{x}|) \sum_{\pi \in P_{l,m}} \frac{l!}{\pi!} \frac{1}{r^l x_1^{2l-m}} e_{\pi,l}(\b{x})\notag\\
&=  \frac{1}{r^l x_1^{2l-1}}\sum_{m=1}^l \partial^m_{x_1}G(|\b{x}|) \sum_{\pi \in P_{l,m}} \frac{l!}{\pi!} x_1^{m-1} e_{\pi,l}(\b{x}).
\label{eq:glrfinalsimp}
\end{align}
Let
\begin{equation}
\label{eq:slm-def}
s_{l,m}(\b{x})=\sum_{\pi \in P_{l,m}} \frac{l!}{\pi!} x_1^{m-1} e_{\pi,l}(\b{x}),
\end{equation}
where $s_{l,m}(\b{x})$ is a polynomial in $x_1, \dots, x_d$ (since $m \geq 1$ ensures $x_1^{m-1}$ has a non-negative exponent). Substitute \eqref{eq:glrfinalsimp} into \eqref{eq:partialqxbulky} to obtain
\[
\partial^{\b{q}}_{\b{x}} G(|\b{x}|) =\sum_{k=1}^{|\b{q}|} \sum_{l=0}^k \frac{(-1)^{k-l}\risingfac{l}{2(k-l)}}{(k-l)!(2 r)^{2k-l}} \frac{1}{r^l x_1^{2l-1}}\sum_{m=1}^l \partial^m_{x_1}G(|\b{x}|) s_{l,m} (\b{x}) b_{\b{q},k}(\b{x}) .
\]
Combining powers of $r$ gives
\begin{equation}
\label{eq:final1}
\partial^{\b{q}}_{\b{x}} G(|\b{x}|) =\sum_{k=1}^{|\b{q}|}\frac{b_{\b{q},k}(\b{x}) }{(2 r)^{2k}} \sum_{l=0}^k \frac{(-1)^{k-l}\risingfac{l}{2(k-l)} \, 2^l}{(k-l)!} \frac{1}{x_1^{2l-1}}   \sum_{m=1}^l \partial^m_{x_1}G(|\b{x}|) s_{l,m} (\b{x}),
\end{equation}
recalling the definitions of $b_{\b q,k}$ from \eqref{eq:bqk-def} and $s_{l,m}$ from \eqref{eq:slm-def}.

Upon multiplication of the factor $x_1^{2|\b{q}|-1} r^{2|\b{q}|}$, \eqref{eq:final1} contains no rational coefficients, since all powers of $r$ present are even and $b_{\b{q},k}(\b{x}), s_{l,m} (\b{x}), w_{j}(\b{x})$ are multivariate polynomials in $x_1, \dots, x_d$. Thus it is possible to rewrite \eqref{eq:PDE} as an ODE in $x_1$ with coefficients that are multivariate polynomials in $x_1, \dots, x_d$, if we normalize by multiplying by $x_1^{2|\b{q}|-1} r^{2|\b{q}|} $.
\end{proof}

\subsection{ODE to Large-$|x_1|$ Recurrence}
\label{sec:large-x1}

The goal of this section is to show that if we have an expression with coefficients in $\C[x_1, \dots, x_d]$ and partial derivatives of $G(\b{x})$ with respect to $x_1$, then we can derive a recurrence for the derivatives of $G(\b{x})$ with respect to $x_1$. Foreshadowing some later developments, we call this recurrence the large-$|x_1|$ recurrence because we observe it is numerically stable in a cone around the $x_1$-axis (see empirical results in Section~\ref{sec:error-modeling} and Figure~\ref{fig:my_label}).

\begin{myexample}{Large-$|x_1|$ Recurrence for Laplace in 2D}{laplace-large-x1}
We resume from Example~\ref{myex:laplace-ode} by recalling the ODE \eqref{eq:normalized-laplace} satisfied by the first and second derivatives of $G(|\b{x}|)$ with respect to $x_1$:
\[
(x_1^3 + x_1 x_2^2)\partial^2_{x_1} G(|\b{x}|) + (x_1^2 - x_2^2) \partial_{x_1} G(|\b{x}|)  = 0.
\]
Our goal is to generalize this ODE to arbitrary-order derivatives. To this end we consider taking one derivative of \eqref{eq:normalized-laplace} (i.e. applying $\partial^1_{x_1}$):
\[
(x_1^3 + x_1 x_2^2)  \partial^3_{x_1} G(|\b{x}|) + 4x_1^2 \partial^2_{x_1} G(|\b{x}|) + 2x_1 \partial_{x_1} G(|\b{x}|) = 0.
\]
With two derivatives (i.e. applying $\partial^2_{x_1}$):
\[
(x_1^3 + x_1 x_2^2) \partial^4_{x_1} G(|\b{x}|) + (7x_1^2 + x_2^2) \partial^3_{x_1} G(|\b{x}|) + 10 x_1 \partial^2_{x_1} G(|\b{x}|) + 2 \partial_{x_1} G(|\b{x}|) = 0.
\]
Lastly, an inductive argument shows the following order-parametric ODE holds (i.e. applying $\partial^n_{x_1}$):
\begin{equation}
\label{eq:laplace2dlarge-x1}
\begin{aligned}
(x_1^3 + x_1 x_2^2) \partial^{n+2}_{x_1} G(|\b{x}|) + ([3n + 1] x_1^2 + [n-1] x_2^2) \partial^{n+1}_{x_1} G(|\b{x}|)\\ +([3n^2 - n] x_1)\partial^{n}_{x_1} G(|\b{x}|) + n(n-1)^2 \partial^{n-1}_{x_1} G(|\b{x}|)= 0.
\end{aligned}
\end{equation}
Observe that \eqref{eq:laplace2dlarge-x1}  provides a recurrence for the derivatives of $G$ that we call the large-$|x_1|$ recurrence for the Laplace PDE in 2D.
\end{myexample}

\begin{figure}[!ht]
\centering
\begin{tikzpicture}[scale=0.8]
\fill [color=black!20](0,0) -- (70:5) -- ++(2.5,0) -- ($(-70:5) + (2.5,0) $)-- ++(-2.5,0) --cycle;
\fill [color=black!20](0,0) -- (110:5) -- ++(-2.5,0) -- ($(-110:5) + (-2.5,0) $)-- ++(2.5,0) --cycle;
\draw [->, thick] (-5,0) -- (5,0) node [anchor=north,pos=1] {${x}_1$};
\draw [->, thick] (0,-5) -- (0,5) node [anchor=west,pos=1] {${x}_2$};
\foreach\rad in {1,2,3,4}
  \draw  [ dashed](0,0) circle (\rad);
\node [anchor=south] at (0,-4) {Unstable};
\node [anchor=west] at (2,2) {Stable};

\draw[ thick, blue] (0,4.5) cos (0.5,3);
\draw[thick, blue] (0,4.5) cos (-0.5,3);
\node at (0.25, 3)[circle,fill,inner sep=1pt]{};
\draw [->, thick] (-1, 3) -- (1,3);
\node [anchor=north] at (0.25,3) {$\b{x}$};

\end{tikzpicture}

\caption{The shaded region is where the large-$|x_1|$ recurrence is numerically stable. Since $G(|\b{x}|)$ is an even function along $\hat{x}_1$, the odd-even derivatives fluctuate in magnitude near $x_1 = 0$.}
\label{fig:my_label}
\end{figure}

In generalizing Example~\ref{myex:laplace-large-x1} to arbitrary $G$, consider the order $a$ ODE ($a \in \N$) satisfied by $G(|{\b x}|)$
\begin{equation}
\label{eq:ode}
\sum_{i=0}^a l_i({\b x}) \partial^i_{x_1} G(|{\b x}|) = 0 , \qquad l_a(\b{x}) \ne 0,
\end{equation}
where $l_i({\b x})$ is a multivariate polynomial in $x_1, \dots, x_d$. Then expand $l_i(\b{x})$ in $x_1$ to obtain
\begin{equation}
\label{eq:odecoeffexp}
l_i({\b x}) = q_{i0}(x_2, \dots, x_d) + \dots +  q_{ih}(x_2, \dots, x_d) x_1^h = \sum_{j=0}^{h} q_{ij}(x_2, \dots, x_d) x_1^j,
\end{equation}
where 
$h$ is the highest power of $x_1$ present in the ODE, and $q_{ij}(x_2, \dots, x_d)$ is a multivariate polynomial in $x_2, \dots, x_d$, possibly 0. Substituting \eqref{eq:odecoeffexp} into \eqref{eq:ode} gives
\begin{equation}
\label{eq:odeasmono}
\sum_{i=0}^a  \sum_{j=0}^{h} q_{ij}(x_2, \dots, x_d) x_1^j \partial^i_{x_1} G(|{\b x}|) = 0.
\end{equation}
Let $n \in \N_0$ and apply $\partial^n_{x_1}$ to both sides of \eqref{eq:odeasmono}
\[
\sum_{i=0}^a  \sum_{j=0}^{h} q_{ij}(x_2, \dots, x_d) \partial^n_{x_1} \left( x_1^j \partial^i_{x_1} G(|{\b x}|) \right) = 0.
\]
Applying \eqref{eq:productrule1} gives
\begin{equation}
  \sum_{i=0}^a  \sum_{j=0}^{h} \sum_{l=0}^j
  \underbrace{q_{ij}(x_2, \dots, x_d) \left(\prod^{l-1}_{k=0}(n-k)\right) \binom{j}{l} x_1^{j-l}}_{u_{ijl}(\b x):=} \partial^{n-l+i}_{x_1} G(|{\b x}|)= 0,
\end{equation}
which gives an expression for a recurrence of derivatives of $G(|\b{x}|)$ with respect to $x_1$, summarized as
\begin{equation}
\label{eq:large-x1-recur}
\sum_{i=0}^a  \sum_{j=0}^{h}  \sum_{l=0}^j u_{ijl}(\b{x}) \partial^{n-l+i}_{x_1} G(|{\b x}|)= 0 \qquad  (n \in \N_{\geq 0}).
\end{equation}
We call \eqref{eq:large-x1-recur} the \emph{large-$|x_1|$ recurrence}. The order of the recurrence can at most be bounded by $a+h$, where $a$ is the ODE order and $h$ is the highest power of $x_1$ present in \eqref{eq:ode}. 


\subsection{Small-$|x_1|$ Expansion}
\label{sec:small-x1}
The large-$|x_1|$ recurrence \eqref{eq:large-x1-recur} incurs rounding error when
\begin{equation}
\label{eq:cone-small-x1}
\frac{|x_1|}{\overline{x}} \ll 1,
\qquad\text{where}\qquad
\overline{x} = \sqrt{x_2^2 + \dots + x_d^2}.
\end{equation}
$G$, in the region where \eqref{eq:large-x1-recur} holds, is an even function of $x_1$. As a result, in a cone around $x_1=0$, as shown in Figure~\ref{fig:my_label}, odd-order derivatives will almost vanish, while even-order derivatives will not, leading to considerable differences in their magnitude.
In this region, the recurrence \eqref{eq:large-x1-recur} may thus involve terms of strongly varying magnitude, inviting rounding error. Indeed, this rounding error is large enough to entirely destroy the accuracy of the computed derivatives at moderate orders. To obviate this issue, we develop what we call the small-$|x_1|$ expansion.

\begin{myexample}{Small-$|x_1|$ Expansion for Laplace in 2D}{small-x1-laplace}
We continue from Example~\ref{myex:laplace-large-x1}. Suppose we start with the large-$|x_1|$ recurrence for the Laplace PDE in 2D
\[
\begin{aligned}
(x_1^3 + x_1 x_2^2) \partial^{n+2}_{x_1} G(|\b{x}|) + ([3n + 1] x_1^2 + [n-1] x_2^2) \partial^{n+1}_{x_1} G(|\b{x}|)\\ +([3n^2 - n] x_1)\partial^{n}_{x_1} G(|\b{x}|) + n(n-1)^2 \partial^{n-1}_{x_1} G(|\b{x}|)= 0.
\end{aligned}
\]
We perform the substitution $x_1 = 0$ to get
\begin{equation}
\label{eq:large-x1-recur-laplace}
\left(\partial^n_{x_1} G(|\b{x}|) \right)|_{x_1 = 0} = \frac{-(n-2)(n-1)\left(\partial^{n-2}_{x_1} G(|\b{x}|)\right)_{x_1=0}}{x_2^2}.
\end{equation}
Since $G(|\b x|)$ is an even function, the substitution eliminates `odd-order' terms, lowering computational expense. We call \eqref{eq:large-x1-recur-laplace} the small-$|x_1|$ recurrence for the Laplace PDE in 2D.

Note that the small-$|x_1|$ recurrence is NOT in itself a recurrence for the derivatives of $G(|\b{x}|)$, it is a recurrence for the derivatives evaluated at $x_1 = 0$: $\left(\partial^n_{x_1} G(|\b{x}|) \right)|_{x_1 = 0}$. To actually compute the derivatives, we use a Taylor expansion of $\partial^n_{x_1} G(|\b{x}|)$, centered around $x_1=0$, with respect to $x_1$.  Using a truncation order of, say $p_{\text{small-$|x_1|$}} = 2$, we obtain a third-order approximation of the derivative:\[
\partial^n_{x_1} G(|\b{x}|) = \partial^n_{x_1}G(|\b{x}|)|_{x_1=0}+  \partial^{1}_{x_1} \partial^n_{x_1}G(|\b{x}|)|_{x_1=0} x_1 + \partial^{2}_{x_1} \partial^n_{x_1}G(|\b{x}|)|_{x_1=0} \frac{x_1^2}{2} + O(x_1^3)
\]
Because $G(|\b x|)$ is an even function, all the `net-odd-order' (i.e., considering $n$) terms are zero. We can thus avoid evaluating the (numerically troublesome) odd-order derivatives, avoiding a source of inaccuracy that beset our large-$|x_1|$ recurrence.
\end{myexample}

The first step is to recognize that the quantities of interest $\{\partial^n_{x_1} G(|\b{x}|)\}_{n=0}^{\infty}$ can be expanded in a Taylor series around $x_1 = 0$:
\begin{equation}
\label{eq:taylor-rep-derivs}
\partial^n_{x_1} G(|\b{x}|) = \sum_{s=0}^{\infty} \left( \partial^{s+n}_{x_1} G(|\b{x}|) \right)|_{x_1=0} \frac{x_1^s}{s!}.
\end{equation}
If we have a numerically stable recurrence for $\{(\partial^k_{x_1} G(|\b{x}|))|_{x_1=0}\}|_{k=0}^{\infty}$ with negligible rounding error, then we can compute \eqref{eq:taylor-rep-derivs} to high accuracy cheaply by truncating our representation at $p_{\text{small-$|x_1|$}} \in \N$:
\begin{equation}
\label{eq:taylor-trunc}
\partial^n_{x_1} G(|\b{x}|) = \sum_{s=0}^{p_{\text{small-$|x_1|$}}} \left( \partial^{s+n}_{x_1} G(|\b{x}|) \right)|_{x_1=0} \frac{x_1^s}{s!} + \left(\partial^{p_{\text{small-$|x_1|$}}+n+1}_{x_1} G(|\b{x}|)\right)_{\b{x}=\b{t}_{\xi}} \frac{x_1^{p_{\text{small-$|x_1|$}}+1}}{(p_{\text{small-$|x_1|$}}+1)!}
\end{equation}
where $\b{t}_{\xi}= (\xi, x_2, \dots, x_d)$. The recurrence for $\{(\partial^i_{x_1} G(|\b{x}|))_{x_1=0}\}|_{i=0}^{\infty}$ can be obtained by substitution of $x_1=0$ into our large-$|x_1|$ recurrence \eqref{eq:large-x1-recur}, resulting in
\begin{equation}
\label{eq:small-x1-1}
\sum_{i=0}^a  \sum_{j=0}^{h}  \sum_{l=0}^j u_{ijl}(\b{x})|_{x_1=0} \left(\partial^{n-l+i}_{x_1} G(|{\b x}|) \right)|_{x_1=0}= 0 \qquad  (n \in \N_{0}).
\end{equation}
We call \eqref{eq:small-x1-1} the \emph{small-$|x_1|$ recurrence}. Note that the small-$|x_1|$ recurrence is \emph{not} a recurrence for the derivatives of $\{\partial^i_{x_1} G(|\b{x}|)\}_{i=0}^\infty$, but instead a recurrence for the sequence: $\{\left(\partial^i_{x_1} G(|\b{x}|) \right)_{x_1=0}\}_{i=0}^\infty$. Once elements in the sequence $\{\left(\partial^i_{x_1} G(|\b{x}|) \right)_{x_1=0}\}_{i=0}^\infty$ are computed, they can be substituted into $\eqref{eq:taylor-trunc}$ to obtain an approximation to the set of derivatives we want: $\{\partial^i_{x_1} G(|\b{x}|)\}_{i=0}^\infty$, an approximation that has a well-understood truncation error profile with respect to $p_{\text{small-$|x_1|$}}$. 

Since $G(|\b{x}|)$ is an even function of $x_1$, $(\partial^k_{x_1} G(|\b{x}|))|_{x_1=0}=0$ when $k$ is odd. As a result, the small-$|x_1|$ recurrence avoids the effect of even-odd oscillations. We observe that the small-$|x_1|$ expansion is numerically stable for all orders with negligible rounding error (see Section~\ref{sec:error-modeling} for error modeling). Furthermore, the small-$|x_1|$ recurrence has recurrence order less than or equal to that of the large-$|x_1|$ recurrence. This is because in the small-$|x_1|$ recurrence, odd-order terms must be zero, and therefore the recurrence order is typically strictly less than that of the large-$|x_1|$ recurrence.

\subsection{Recurrence Error Modeling}
\label{sec:error-modeling}
So far, we have developed two approaches for the evaluation of derivatives that can be produced for an arbitrary kernel: one that uses the large-$|x_1|$ recurrence and one that uses the small-$|x_1|$ expansion. We present error models for both approaches in this section.

\begin{figure}[h]
\centering
\includegraphics{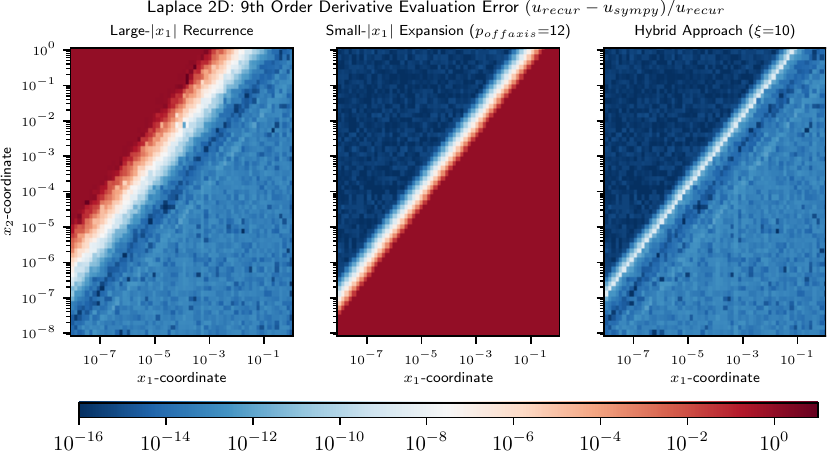}
\caption{The relative error in evaluating the 9th derivative of the Green's function for the Laplace PDE in 2D using recurrences ($u_{recur}$) versus using the symbolic calculator SymPy ($u_{sympy}$) \cite{10.7717/peerj-cs.103} for different evaluation points. The large-$|x_1|$/small-$|x_1|$ recurrence alternates use of the large-$|x_1|$ and small-$|x_1|$ recurrence above/below the line $x_2 = \xi x_1$ for $\xi =10$.}
\label{figure:Laplace2D9thOrder}
\end{figure}

\begin{figure}[h]
\centering
\includegraphics{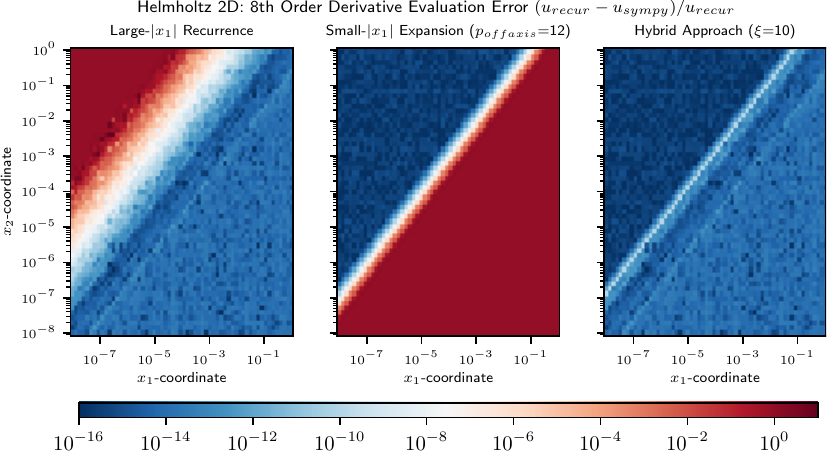}
\caption{The relative error in evaluating the 8th derivative of the Green's function for the Helmholtz PDE in 2D using recurrences ($u_{recur}$) versus using the symbolic calculator SymPy ($u_{sympy}$) \cite{10.7717/peerj-cs.103} for different evaluation points. The large-$|x_1|$/small-$|x_1|$ recurrence alternates use of the large-$|x_1|$ and small-$|x_1|$ recurrence above/below the line $x_2 = \xi x_1$ for $\xi = 10$.}
\label{figure:Helmholtz2D8thOrder}
\end{figure}

\begin{figure}[h]
\centering
\includegraphics{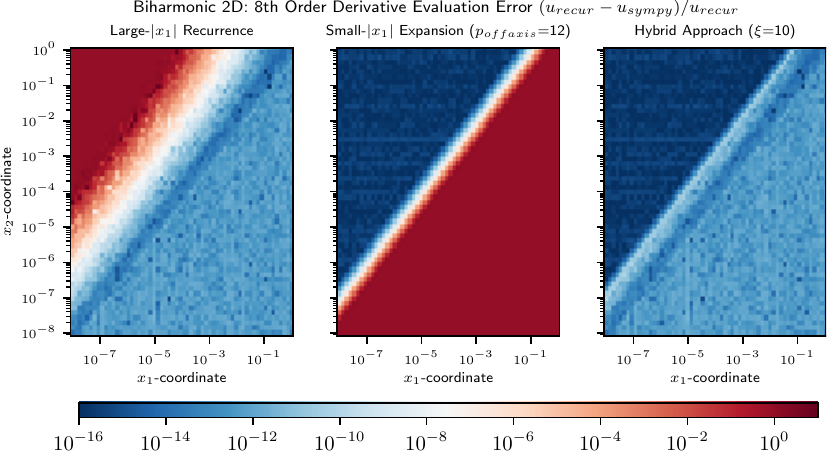}
\caption{The relative error in evaluating the 8th derivative of the Green's function for the Biharmonic PDE in 2D using recurrences ($u_{recur}$) versus using the symbolic calculator SymPy ($u_{sympy}$) \cite{10.7717/peerj-cs.103} for different evaluation points. The large-$|x_1|$/small-$|x_1|$ recurrence alternates use of the large-$|x_1|$ and small-$|x_1|$ recurrence above/below the line $x_2= \xi x_1$ for $\xi = 10$.}
\label{figure:Biharmonic2D7thOrder}
\end{figure}

\subsubsection{Large-$|x_1|$ Recurrence Error Modeling}
\label{sec:large-x1-error-model}
The large-$|x_1|$ recurrence in \eqref{eq:large-x1-recur} incurs uncontrolled rounding error when $|x_1| \ll \overline{x}$. See Figures \ref{figure:Laplace2D9thOrder}--\ref{figure:Biharmonic2D7thOrder} for concrete examples of this phenomenon in 2D, specifically how the large-$|x_1|$ recurrence performs poorly above the line $x_2 = \xi x_1$ for $\xi = 10$. 

We begin by obtaining an approximate bound for the recurrence error at each step. Recall that the recurrence formula for our derivatives is of the general form
\begin{align}
\label{eq:orderrderv}
\partial^n_{x_1} G(|\b{x}|)
&= a_{n-1}(\b{x}) \partial^{n-1}_{x_1} G(|\b{x}|) + \dots + a_{n-r}(\b{x})\partial^{n-r}_{x_1} G(|\b{x}|)\\
&= b_{n-1}(\b{x}) + \dots + b_{n-r}(\b{x}),
\end{align}
where $r \in \N$ is the recurrence order and $b_{n-i}(\b{x}) =  a_{n-i}(\b{x}) \partial^{n-i}_{x_1} G(|\b{x}|)$ $(i \in \{1, \dots, r\})$. Let $\oplus$ represent addition with rounding. Then the relative rounding error present in evaluation of $\partial^n_{x_1} G(|\b{x}|)$ in a single recurrence step is given by
\begin{equation}
\label{eq:crude-additions}
E:=|\partial^n_{x_1} G(|\b{x}|) - (b_{n-1}(\b{x}) \oplus \dots \oplus b_{n-r}(\b{x}))|/|b_{n-1}(\b{x}) \oplus \dots \oplus b_{n-r}(\b{x})|
\end{equation}
for some $\b{x} \in \R^d$. Under certain assumptions discussed below, an approximate bound for this is given by
\begin{equation}
\label{eq:crude-error}
E \lessapprox \max_{i,j \in \{1, \dots, r\}} \frac{|b_{n-i}(\b{x})|}{|b_{n-j}(\b{x})|}.
\end{equation}
To derive this approximate bound, suppose we are given two (already rounded) numbers $x, y\in\mathbb R$ where $|y| \ll |x|$.  Then we assume that the absolute rounding error present in $x \oplus y$ can be coarsely bounded by
\[
|(x\oplus y) - (x+y)|  \lessapprox |y|
\]
and thus the relative rounding error can be coarsely approximated by
\[
\frac{|(x\oplus y) - (x+y)|}{|x+y|}  \lessapprox \frac{|y|}{|x+y|}\approx \frac{|y|}{|x|}.
\]

Applying the two-term rounding estimate to the summation
$b_{n-1}(\b{x}) \oplus \cdots \oplus b_{n-r}(\b{x})$ in \eqref{eq:crude-additions}, we take a pairwise worst-case view: we assume the
overall relative rounding error in this recurrence step is approximately upper-bounded by the largest
relative error that could occur in adding any two contributions, which leads to the estimate
\[E \lessapprox \max_{i,j\in\{1,\dots,r\}} |b_{n-i}(\b{x})|/|b_{n-j}(\b{x})|\]
in \eqref{eq:crude-error}. We note that this merely estimates the maximal rounding error
modeled via \eqref{eq:crude-error} a result of pairwise addition and neglects subsequent
additions occurring on intermediate results.

We now apply the approximate single-step estimate \eqref{eq:crude-error} to the 2D Laplace
large-$|x_1|$ recurrence by matching the generic notation \eqref{eq:orderrderv} to a specific
recurrence instance.  Starting from the large-$|x_1|$ recurrence in Example~2 \eqref{eq:laplace2dlarge-x1},
we substitute $n=7$, which is the instance that produces $\partial_{x_1}^9 G(|\b{x}|)$ as a linear
combination of $\partial_{x_1}^8 G$, $\partial_{x_1}^7 G$, and $\partial_{x_1}^6 G$:
\begin{equation}
\label{eq:step1}
\partial^{9}_{x_1} G(|\b{x}|)
= a_{8}(\b{x})\,\partial^{8}_{x_1} G(|\b{x}|)
+ a_{7}(\b{x})\,\partial^{7}_{x_1} G(|\b{x}|)
+ a_{6}(\b{x})\,\partial^{6}_{x_1} G(|\b{x}|).
\end{equation}
Comparing \eqref{eq:step1} with \eqref{eq:orderrderv} shows that this recurrence step has order $r=3$,
with addends
\begin{align*}
b_{8}(\b{x}) &:= a_{8}(\b{x})\,\partial^{8}_{x_1} G(|\b{x}|),\\
b_{7}(\b{x}) &:= a_{7}(\b{x})\,\partial^{7}_{x_1} G(|\b{x}|),\\
b_{6}(\b{x}) &:= a_{6}(\b{x})\,\partial^{6}_{x_1} G(|\b{x}|).
\end{align*}
Therefore, \eqref{eq:crude-error} specializes to
\[
E \lessapprox \max_{p,q\in\{6,7,8\}} \frac{|b_p(\b{x})|}{|b_q(\b{x})|}
\]
for the evaluation of \eqref{eq:step1}.  To evaluate this quantity in the regime where the
large-$|x_1|$ recurrence exhibits loss of accuracy, we set $\b{x}=(x_1,\overline{x})$ and substitute
the exact symbolic expressions for $\partial_{x_1}^6G$, $\partial_{x_1}^7G$, and $\partial_{x_1}^8G$
into $b_6,b_7,b_8$.  Simplifying and expanding for $|x_1|/\overline{x}\to 0$ shows that $b_6$ and $b_8$
scale like $O((\overline{x}/x_1)^{-8})$ while $b_7$ is smaller, $O((\overline{x}/x_1)^{-10})$, so the
maximum is dominated by ratios such as $|b_6|/|b_7|$ (or $|b_8|/|b_7|$). In particular,
\begin{equation}
\label{eq:large-x1-9}
\max_{p,q\in\{6,7,8\}} \frac{|b_p(\b{x})|}{|b_q(\b{x})|}
=O\!\left(\left(\frac{|x_1|}{\overline{x}}\right)^{-2}\right)
\end{equation}
as $|x_1|/\overline{x} \to 0$.
\begin{figure}[h]
\centering
\includegraphics{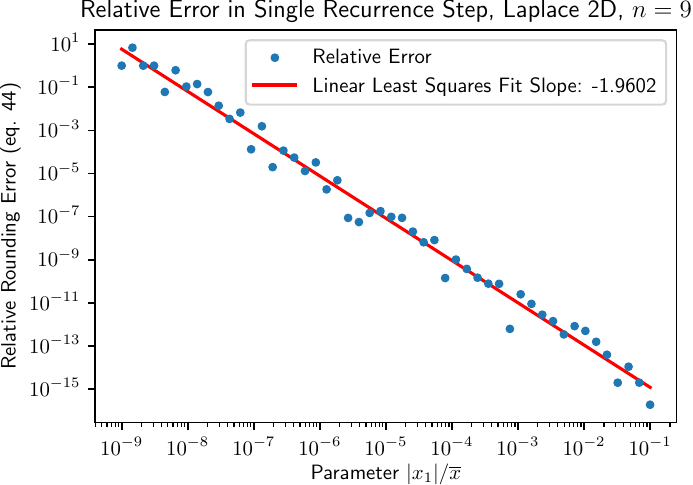}
\caption{We plot the relative error due to rounding in the large-$|x_1|$ recurrence, versus $|x_1|/\overline{x}$, for a single recurrence step, assuming our preceding terms are given exactly with no rounding error. We sample a set of 5 random points for each $|x_1|/\overline{x}$ value.}
\label{fig:llf}
\end{figure}
The error model in \eqref{eq:large-x1-9} is for a single recurrence step, however it implies that the ``compound" rounding error from the application of multiple recurrence steps should also only depend primarily on $|x_1|/\overline{x}$.

\textbf{Empirical validation.} Figure~\ref{figure:Laplace2D9thOrder} shows a heat-map for the error in the large-$|x_1|$ recurrence. The observed behavior appears to support that $|x_1|/\overline{x}$ is the primary factor determining the error.

Next, we carry out an experiment to validate the exponent in the asymptotics of \eqref{eq:large-x1-9} as $|x_1|/\overline{x} \to \infty$. In the experiment, carried out for the Laplace PDE in 2D, we evaluate $\partial^n_{x_1} G(|\b{x}|)$ for $n = 9$, assuming $\partial^{n-1}_{x_1} G(|\b{x}|), \dots, \partial^{n-r}_{x_1} G(|\b{x}|)$ are known exactly. To this end, we generate a random set of points in $\R^2$ such that they have values of $|x_1|/\overline{x} = 10^0, 10^1, \dots, 10^9$. We choose five random points per chosen value of $|x_1|/\overline{x}$. We then plot the relative error given by \eqref{eq:crude-additions} versus the parameter $|x_1|/\overline{x}$, and label the slope for the linear least squares fit. The results are shown in Figure~\ref{fig:llf}, which exhibits good agreement with the exponent in \eqref{eq:large-x1-9}.

\subsubsection{Error Modeling for the Small-$|x_1|$ Expansion}
As we have seen, rounding error appears to dominate in the case of the large-$|x_1|$ recurrence. For the small-$|x_1|$ expansion, by contrast,
truncation error turns out to be the dominant source of error.

Using a remainder for the Taylor expansion in \eqref{eq:taylor-trunc} yields an asymptotic estimate for the truncation error in that expansion:
\begin{equation}
\label{eq:asymptotic-error-71}
\left|\partial^n_{x_1} G(|\b{x}|) - \sum_{s=0}^{p_{\text{small-$|x_1|$}}} \left( \partial^{s+n}_{x_1} G(|\b{x}|) \right)|_{x_1=0} \frac{x_1^s}{s!} \right| \leq \max_{0\leq \xi\leq x_1} \left|\left(\partial^{p_{\text{small-$|x_1|$}}+n+1}_{x_1}G(|\b{x}|)\right)_{\b{x}=\b{t}_\xi}\right| \frac{ x_1^{p_{\text{small-$|x_1|$}}+1}}{(p_{\text{small-$|x_1|$}}+1)!},
\end{equation}
where $\b{t}_\xi = (\xi, x_2, \dots, x_d)$. Realize that in the sum
\begin{equation}
\label{eq:series-under-consider}
\sum_{s=0}^{p_{\text{small-$|x_1|$}}} \left( \partial^{s+n}_{x_1} G(|\b{x}|) \right)|_{x_1=0} \frac{x_1^s}{s!}
\end{equation}
only the terms where $s+n$ are even are non-zero, based on the same even-odd effect described above. This implies that we can assume WLOG that $p_{\text{small-$|x_1|$}}$ is chosen such that $p_{\text{small-$|x_1|$}}+n$ is even. This ensures that the final term in \eqref{eq:series-under-consider} is non-zero. Furthermore, this implies that $p_{\text{small-$|x_1|$}}+n+1$ is odd, which is relevant since the term
\[
\left(\partial^{p_{\text{small-$|x_1|$}}+n+1}_{x_1}G(|\b{x}|)\right)_{\b{x}=\b{t}_\xi}
\]
appears in the RHS of \eqref{eq:asymptotic-error-71}. We now detail a list of assumptions that can be used to construct a more specific relative error bound for the evaluation of $\partial^n_{x_1} G(|\b{x}|)$ using the small-$|x_1|$ recurrence.

In bounding truncation error, we make use of some assumptions on the growth of Green's function derivatives.
\begin{assumption}
\label{ass:deriv-growth-odd}
Let $\b{x} \in \R^d \setminus \{\b{0}\}$. If $|x_1|/\overline{x} < 1$, $\mu \in \N$ is odd, then there exists an $M_\mu \in \R_{>0}$ such that
\begin{equation}
\max_{0\leq \xi\leq x_1} \left|\left(\partial^{\mu}_{x_1}G(|\b{x}|)\right)_{\b{x}=\b{t}_\xi}\right| \leq M_\mu\frac{|x_1|}{\overline{x}^{\mu+1}}.
\end{equation}
\end{assumption}
\begin{assumption}
\label{ass:deriv-growth-odd-even}
Let $\b{x} \in \R^d \setminus \{\b{0}\}$. If $|x_1|/\overline{x} < 1$ and $\mu \in \N$ is odd then there exists an $m_\mu \in \R_{>0}$ such that
\[
\max_{0\leq \xi\leq x_1} \left|\left(\partial^{\mu}_{x_1}G(|\b{x}|)\right)_{\b{x}=\b{t}_\xi}\right| \geq m_\mu\frac{|x_1|}{\overline{x}^{\mu+1}}.
\]
If $\nu \in \N$ is even, $|x_1|/\overline{x} < 1$ then there exists an $m_\nu \in \R_{>0}$ such that
\[
\max_{0\leq \xi\leq x_1} \left|\left(\partial^{\nu}_{x_1}G(|\b{x}|)\right)_{\b{x}=\b{t}_\xi}\right| \geq m_\nu\frac{1}{\overline{x}^{\nu}}.
\]
\end{assumption}

\begin{assumption}
\label{ass:deriv-growth-cubed-odd}
Let $\b{x} \in \R^d \setminus \{\b{0}\}$. If $|x_1|/\overline{x} < 1$, $\mu \in \N$ is odd, then there exists an $M_\mu \in \R_{>0}$ so that
\[
\max_{0\leq \xi\leq x_1} \left|\left(\partial^{\mu}_{x_1}G(|\b{x}|)\right)_{\b{x}=\b{t}_\xi}\right| \leq M_\mu\frac{|x_1|^{3}}{\overline{x}^{\mu+1}}.
\]
for all $\b{x} \in \R^d$.
\end{assumption}

\begin{assumption}
\label{ass:deriv-growth-cubed-odd-even}
Let $\b{x} \in \R^d \setminus \{\b{0}\}$. If $|x_1|/\overline{x} < 1$ and $\mu \in \N$ is odd then there exists an $m_\mu \in \R_{>0}$ such that
\[
\max_{0\leq \xi\leq x_1} \left|\left(\partial^{\mu}_{x_1}G(|\b{x}|)\right)_{\b{x}=\b{t}_\xi}\right| \geq m_\mu\frac{|x_1|^3}{\overline{x}^{\mu+1}}.
\]
If $\nu \in \N$ is even, and $|x_1|/\overline{x} < 1$, then there exists an $m_\nu \in \R_{>0}$ such that
\[
\max_{0\leq \xi\leq x_1} \left|\left(\partial^{\nu}_{x_1}G(|\b{x}|)\right)_{\b{x}=\b{t}_\xi}\right| \geq m_\nu\frac{|x_1|^2}{\overline{x}^{\nu}}.
\]
\end{assumption}
Figures \ref{fig:fig10} and \ref{fig:fig11} in Appendix~\ref{app:assumption-numerics} support the assertion that Assumption~\ref{ass:deriv-growth-odd} and \ref{ass:deriv-growth-odd-even} hold for Laplace/Helmholtz PDE in 2D and that Assumption~\ref{ass:deriv-growth-cubed-odd} and \ref{ass:deriv-growth-cubed-odd-even} hold for the Biharmonic PDE in 2D, by empirically validating the assumption when $\mu=5$ and $\nu=6$.
When Assumption~\ref{ass:deriv-growth-odd} holds, \eqref{eq:asymptotic-error-71} tells us that the small-$|x_1|$ absolute error for Laplace/Helmholtz 2D can be written in the form:
\begin{equation}
\label{eq:abs-error}
\max_{0\leq \xi\leq x_1} \left|\left(\partial^{p_{\text{small-$|x_1|$}}+n+1}_{x_1}G(|\b{x}|)\right)_{\b{x}=\b{t}_\xi}\right| \frac{ x_1^{p_{\text{small-$|x_1|$}}+1}}{(p_{\text{small-$|x_1|$}}+1)!} \leq   \frac{M_{p_{\text{small-$|x_1|$}}+n+1}}{(p_{\text{small-$|x_1|$}}+1)!}\frac{|x_1|^{p_{\text{small-$|x_1|$}}+2}}{\overline{x}^{n+p_{\text{small-$|x_1|$}}+2}}.
\end{equation}
When $n$ is even, Assumption~\ref{ass:deriv-growth-odd-even} implies that the relative error when $|x_1|/\overline{x} < 1$ is bounded and given by:
\[
\left|\partial^n_{x_1} G(|\b{x}|) - \sum_{s=0}^{p_{\text{small-$|x_1|$}}} \left( \partial^{s+n}_{x_1} G(|\b{x}|) \right)|_{x_1=0} \frac{x_1^s}{s!} \right| / \left| \partial^n_{x_1} G(|\b{x}|)  \right| \leq \frac{M_{p_{\text{small-$|x_1|$}}+n+1}}{m_n(p_{\text{small-$|x_1|$}}+1)!} \left(\frac{|x_1|}{\overline{x}}\right)^{p_{\text{small-$|x_1|$}}+2}.
\]
When $n$ is odd, Assumption~\ref{ass:deriv-growth-odd-even} implies that the relative error when $|x_1|/\overline{x} < 1$ is bounded and given by:
\begin{equation}
\label{eq:odd-rel-error}
\left|\partial^n_{x_1} G(|\b{x}|) - \sum_{s=0}^{p_{\text{small-$|x_1|$}}} \left( \partial^{s+n}_{x_1} G(|\b{x}|) \right)|_{x_1=0} \frac{x_1^s}{s!} \right| / \left| \partial^n_{x_1} G(|\b{x}|)  \right| \leq \frac{M_{p_{\text{small-$|x_1|$}}+n+1}}{m_n(p_{\text{small-$|x_1|$}}+1)!} \left(\frac{|x_1|}{\overline{x}} \right)^{p_{\text{small-$|x_1|$}}+1}.
\end{equation}
In either case the relative error bound is solely dependent on the ratio ${|x_1|}/{\overline{x}}$. 

Similar bounds for absolute and relative error can be derived using Assumptions~\ref{ass:deriv-growth-cubed-odd} and~\ref{ass:deriv-growth-cubed-odd-even} for the Biharmonic PDE in 2D.

\subsection{Overall Hybrid Numerical-Symbolic Algorithm}
Given that the dominant error behaviors for both the small-$|x_1|$ expansion and the large-$|x_1|$ recurrence depend on $|x_1|/\overline{x}$, we may use its value as a decision criterion for a hybrid scheme combining the two. For a parameter $\xi \in \R_{> 1}$, we define two regions $\Omega_{\xi,\text{large-$|x_1|$}}, \Omega_{\xi,\text{small-$|x_1|$}}$:
\begin{align*}
\Omega_{\xi,\text{large-$|x_1|$}} &= \{(x_1, \dots, x_d): |x_1|/\overline{x}\geq 1/\xi \},\\
\Omega_{\xi,\text{small-$|x_1|$}} &= \{(x_1, \dots, x_d): |x_1|/\overline{x}< 1/\xi \}.
\end{align*}
Given $\xi$ and $P\in \N$, a hybrid algorithm to compute 
\[
\partial^1_{x_1} G(\b{x}), \dots, \partial^P_{x_1} G(\b{x})
\]
is given in Algorithm \ref{alg:cap}.

\begin{figure}[!ht]
\centering
\begin{tikzpicture}[scale=0.8]
\fill [color=black!20](0,0) -- (70:5) -- ++(2.5,0) -- ($(-70:5) + (2.5,0) $)-- ++(-2.5,0) --cycle;
\fill [color=black!20](0,0) -- (110:5) -- ++(-2.5,0) -- ($(-110:5) + (-2.5,0) $)-- ++(2.5,0) --cycle;
\draw [->, thick] (-5,0) -- (5,0) node [anchor=north,pos=1] {${x}_1$};
\draw [->, thick] (0,-5) -- (0,5) node [anchor=west,pos=1] {${x}_2$};
\foreach\rad in {1,2,3,4}
  \draw  [ dashed](0,0) circle (\rad);
\node [anchor=south] at (0,3) {$\Omega_{\xi,\text{small-$|x_1|$}}$};
\node [anchor=west] at (2,2) {$\Omega_{\xi,\text{large-$|x_1|$}}$};
\end{tikzpicture}

\caption{
$\Omega_{\xi,\text{small-$|x_1|$}}$ and $\Omega_{\xi,\text{large-$|x_1|$}}$, indicating
the regions of space in which each respective component of the hybrid algorithm
is being used.
}
\end{figure}

\begin{algorithm}
\caption{Hybrid Numerical-Symbolic Algorithm for Derivative Computation}\label{alg:cap}
\begin{algorithmic}
\Precomputation{}
\Inputs{
\vspace{-1em}
\begin{itemize}
\item$\mathcal{L}_{\b{x}}$ such that $\mathcal{L}_{\b{x}} G(|{\b x}|)=\sum_{\b{q} \in \mathcal{M}(c)} p_{\b{q}}({\b x}) \partial^{\b{q}}_{{\b x}} G(|{\b x}|) = \delta(\b{x})$ with $\b p_{\b q}\in \C[x_1, \dots, x_d]$
\end{itemize}}
\Outputs{
\vspace{-1em}
\begin{itemize}
\item Large-$|x_1|$ recurrence: $\sum_{i=0}^a  \sum_{j=0}^{h}  \sum_{l=0}^j u_{ijl}(\b{x}) \partial^{n-l+i}_{x_1} G(|{\b x}|)= 0 \qquad  (n \in \N_{\geq 0})$
\item Small-$|x_1|$ recurrence: $\sum_{i=0}^a  \sum_{j=0}^{h}  \sum_{l=0}^j u_{ijl}(\b{x})|_{x_1=0} \left(\partial^{n-l+i}_{x_1} G(|{\b x}|) \right)|_{x_1=0}= 0, \qquad  n \in \N_{0}$
\end{itemize}}
\StateLong{Symbolically translate the PDE that $G(\b{x})$ satisfies into an ODE that $G(|\b{x}|)$ satisfies with respect to $x_1$ following the procedure in Section~\ref{sec:pde-to-ode} (see Example~\ref{myex:laplace-ode}).}
\StateLong{Symbolically translate this ODE into the large-$|x_1|$ recurrence following the procedure in Section~\ref{sec:large-x1} (see Example~\ref{myex:laplace-large-x1}).}
\StateLong{Translate the large-$|x_1|$ recurrence into the small-$|x_1|$ recurrence following the procedure in Section~\ref{sec:small-x1} (see Example~\ref{myex:small-x1-laplace}).}
\Onlinecomputation{}
\Inputs{
\vspace{-1em}
\begin{itemize}
\item Large-$|x_1|$ recurrence: $\sum_{i=0}^a  \sum_{j=0}^{h}  \sum_{l=0}^j u_{ijl}(\b{x}) \partial^{n-l+i}_{x_1} G(|{\b x}|)= 0 \qquad  (n \in \N_{\geq 0})$
\item Small-$|x_1|$ recurrence: $\sum_{i=0}^a  \sum_{j=0}^{h}  \sum_{l=0}^j u_{ijl}(\b{x})|_{x_1=0} \left(\partial^{n-l+i}_{x_1} G(|{\b x}|) \right)|_{x_1=0}= 0, \qquad  n \in \N_{0}$
\item $p_{\text{small-$|x_1|$}} \in \N$, $\xi \in \R_{>1}$, $P \in \N$
\item $\partial^0_{x_1} G(\b{x}), \dots, \partial^a_{x_1} G(\b{x})$, where $a$ is the order of the ODE (these $a+1$ base cases can be computed using a symbolic tool such as SymPy; $a$ depends on the PDE and is typically small relative to $P$)
\end{itemize}
}
\Outputs{
\vspace{-1em}
\begin{itemize}
\item $\partial^1_{x_1} G(\b{x}), \dots, \partial^P_{x_1} G(\b{x})$
\end{itemize}
}
\If{$|x_1|/\overline{x} \geq \frac{1}{\xi}$(i.e. $\b{x} \in \Omega_{\xi,\text{large-$|x_1|$}}$)} 
    \State Compute $\partial^{a+1}_{x_1} G(\b{x}), \dots, \partial^{P}_{x_1} G(\b{x})$ using the large-$|x_1|$ recurrence.
\Else
    \StateLong{Compute $(\partial^0_{x_1} G(\b{x}))|_{x_1=0}, (\partial^2_{x_1} G(\b{x}))|_{x_1=0}, \dots, \left(\partial^{v_1}_{x_1} G(\b{x})\right)|_{x_1=0}$,
      where $v_1 \leq a$ is the largest
      even number less than or equal to $a$.}
    \StateLong{Compute $(\partial^{v_1}_{x_1} G(\b{x}))|_{x_1=0}, (\partial^{v_1+2}_{x_1} G(\b{x}))|_{x_1=0}, \dots, \left(\partial^{v_2}_{x_1} G(\b{x})\right)|_{x_1=0}$ where $v_2 \leq P+p_{\text{small-$|x_1|$}}$  is the largest even number less than or equal to $P+p_{\text{small-$|x_1|$}}$, using the small-$|x_1|$  recurrence.}
    \For{$i \in \{1, \dots, P\}$}
      \State Compute $\partial^i_{x_1} G(\b{x}) = \sum_{k=0,i+k\text{ even}}^{p_{\text{small-$|x_1|$}}} (\partial^{i+k}_{x_1}G(\b{x}))|_{x_1=0}  \frac{x_1^k}{k!} $.
    \EndFor
\EndIf 

\end{algorithmic}
\end{algorithm}

\subsubsection{Error Modeling for the Hybrid Algorithm}
\label{sec:hybrid-error}
As an example, for the case of Laplace PDE in 2D, we can derive a combined error bound for $\partial^9_{x_1} G(|\b{x}|)$, when $p_{\text{small-$|x_1|$}} = 8$, $n =9$, and for $\xi \in \R_{>1}$, using \eqref{eq:odd-rel-error} and \eqref{eq:large-x1-9}:
\begin{equation}
\label{eq:combined-err}
\begin{aligned}
\max\Bigg( \max_{\b{x} \in \Omega_{\xi,\text{small-$|x_1|$}}} \left(\left|\partial^n_{x_1} G(|\b{x}|) - \sum_{s=0}^{p_{\text{small-$|x_1|$}}} \left( \partial^{s+n}_{x_1} G(|\b{x}|) \right)|_{x_1=0} \frac{x_1^s}{s!} \right| / \left| \partial^n_{x_1} G(|\b{x}|)  \right|\right), \\
\max_{\b{x} \in \Omega_{\xi,\text{large-$|x_1|$}}} \left(\frac{3(-\left(\frac{\overline{x}}{x_1}\right)^8+14\left(\frac{\overline{x}}{x_1}\right)^6-14\left(\frac{\overline{x}}{x_1}\right)^2+1)}{10(7\left(\frac{\overline{x}}{x_1}\right)^6-35\left(\frac{\overline{x}}{x_1}\right)^4+21\left(\frac{\overline{x}}{x_1}\right)^2-1)}
\right) \Bigg) \\
\lessapprox  \max_{\b{x}}
\begin{cases} 
      \frac{M_{18}}{m_{9}(p_{\text{small-$|x_1|$}}+1)!} \left({|x_1|/\overline{x}}\right)^9  & \b{x} \in \Omega_{\xi,\text{small-$|x_1|$}}    \\
      C\left(\overline{x}/|x_1| \right)^2 & \b{x} \in  \Omega_{\xi,\text{large-$|x_1|$}}. 
   \end{cases}
\\
\leq \max\left(\frac{M_{18}}{m_{9}(p_{\text{small-$|x_1|$}}+1)!}\left(\frac{1}{\xi}\right)^9, C\xi^2 \right)
\end{aligned}
\end{equation}
for some constant $C$ independent of $\b x$.

When $\b{x} \in  \Omega_{\xi,\text{large-$|x_1|$}}$ \eqref{eq:combined-err} uses the error bound from \eqref{eq:large-x1-9}, and applies it to the last recurrence step only. Hence, we use $\lessapprox$ when $\b{x} \in \Omega_{\xi,\text{large-$|x_1|$}}$, in the sense of Section~\ref{sec:large-x1-error-model}, with the additional stipulation that only the error from the last recurrence step is captured. This latter approximation is justified because we observe that the error in the last recurrence step dwarfs the error in preceding recurrence steps.

Concrete error bounds for other PDEs can be derived in an analogous fashion: for a given PDE, one substitutes the appropriate asymptotic growth assumptions (e.g., Assumptions~\ref{ass:deriv-growth-cubed-odd}--\ref{ass:deriv-growth-cubed-odd-even} for the Biharmonic PDE) and the PDE-specific rounding error model into the same max-of-two-regions structure used in \eqref{eq:combined-err}. The type of dependence on $\xi$ for the combined error is a feature that generalizes beyond the case given by \eqref{eq:combined-err}. In a sense, \eqref{eq:combined-err} resembles other error estimates in scientific computing in that it calls for a balancing of truncation and rounding error as governed by the free parameter $\xi$, much like the error
in numerical differentiation is governed by the choice of a mesh spacing parameter.

\section{Incorporating Recurrences into QBX}
\label{sec:qbx}
A particularly striking setting in which the recurrences introduced above lead to savings of computational cost is the method of quadrature by expansion (QBX) for the evaluation of layer potentials. For example, incorporating their use in the context of QBX can reduce the number of floating point operations compared to a baseline by between one and two orders of magnitude, depending on the PDE and truncation order (cf.~Figures~\ref{fig:cost-2d} and~\ref{fig:cost-3d}). In this section, we briefly review QBX and how to incorporate recurrences into it. We then derive a Cartesian version of Target-Specific QBX~\cite{Siegel_2018} (individually leading to a substantial cost savings), and finally show numerical experiments measuring the computational cost impact of each of the algorithmic steps.

\subsection{QBX Review}
Suppose we have a fundamental solution $K(\b{x}, \b{y})$ where $\b{x} ,\b{y}\in \R^d$. Then if we have a density on $\Gamma$ given by $\sigma:\Gamma\to \mathbb R$, the single layer potential at some $\b{x} \in \Gamma$ is given by:
\begin{equation}
S\sigma(\b{x}) = \int_{\Gamma} K(\b{x}, \b{y})\sigma(\b{y}) d\b{y}.
\end{equation}

Quadrature by expansion (QBX) evaluates layer potentials with singular kernels by using an expansion to approximate the kernel \cite{klockner2013quadrature}. If we let $\b{\nu}(\b{y})$ denote the outwards normal to $\Gamma$ at $\b{y}$, let $r \in \R_{>0}$ be the expansion radius and $p_{QBX}$ is the truncation order, we have from \cite{epstein2013convergence}:
\begin{equation}
\label{eq:qbx}
S\sigma(\b{x}) \approx  \int_{\Gamma} \sum_{i=0}^{p_{\text{QBX}}} \frac{r^i}{i!}\left. \left( \frac{d^i}{dt^i} K(\b{x} + t\b{\nu}, \b{y})  \right)\right|_{t=-r} \sigma(\b{y}) d\b{y}.
\end{equation}
We refer to \cite{epstein2013convergence} for a detailed discussion of the errors incurred in the approximation \eqref{eq:qbx}. Our goal in this work is to speed up the computation of kernel derivatives in \eqref{eq:qbx} using recurrences.

\subsection{Rotations and Line Expansions}
\label{sec:rot}
Let $\b{x},\b{y} \in \R^d$. If $\b{\nu}$ is arbitrarily aligned, then, for $n \in \N$, a computation to evaluate the subexpression
\begin{equation}
\label{eq:bad-chain}
\left. \left( \frac{d^n}{dt^n} K(\b{x} + t\b{\nu}, \b{y})  \right)\right|_{t=-r}
\end{equation}
of \eqref{eq:qbx} can require a large number of floating point operations resulting from the evaluation of chain rule terms. Consider that \eqref{eq:bad-chain} is equivalent to taking $n$ derivatives of the composition $\frac{d^n}{dt^n}  f(\b{g}(t))$ and evaluating it at $t = -r$ where if $\b{z} \in \R^d$
\[
f : \b{z} \mapsto K(\b{z},\b{y})
\]
and
\[
\b{g}: t \mapsto\b{x} + t\b{\nu}.
\]
Thus, without the use of more specific information
on $f: \R^d \to \R$ and $g: \R \to \R^d$, a generic implementation of the symbolic derivative will effectively employ a multi-dimensional Faa Di Bruno formula to compute $\frac{d^n}{dt^n} f(\b{g}(t))$, which will involve a quadratically (in $n$) growing number of cross-terms in its expression.

\begin{figure}[!ht]
\centering
\begin{tikzpicture}[scale=0.65]

\draw [Latex-Latex, dashed] (0,0) -- (-1.41421356237, 1.41421356237) node [anchor=north west,pos=0.74] {$r$};

\draw[fill=black] (0,0) circle (2pt);
\node at (0,0) [minimum size=4pt,inner sep=2pt, anchor=north west,pos=1] {$\b{x}-r\bhat v$};
\draw [->, thick] (-5,0) -- (5,0) node [anchor=north,pos=1] {${x}_1$};
\draw [->] (1.41421356237,1.41421356237) -- (1.41421356237+0.70710678118,1.4142135623+0.70710678118) node [anchor=north,pos=1] {$\bhat v$};
\draw [->, thick] (0,-5) -- (0,5) node [anchor=west,pos=1] {${x}_2$};
\foreach\rad in {2}
  \draw  [ dashed](0,0) circle (\rad);
\draw[blue,thick]
  plot[smooth] coordinates {(-2,3) (1.4142,1.4142) (3,-2)};
\node at (3.3,-1.9) [] {$\Gamma$};

\begin{scope}[rotate around={45:(11,0)}]

\draw [->, thick] (6,0) -- (16,0) node [anchor=north,pos=1] {${x}_1'$};
\draw [->, thick] (11,-5) -- (11,5) node [anchor=west,pos=1] {${x}_2'$};
\foreach\rad in {2}
  \draw  [ dashed](11,0) circle (\rad);
\end{scope}

\draw [->] (11+1.41421356237,1.41421356237) -- (11+1.41421356237+0.70710678118,1.4142135623+0.70710678118) node [anchor=north,pos=1] {$\bhat v'$};

\draw[fill=black] (11,0) circle (2pt);
\node at (10,0)  {\scriptsize $\b{x'}-r\bhat v'$};

\draw[blue,thick]
  plot[smooth] coordinates {(-2+11,3) (1.4142+11,1.4142) (3+11,-2)};
\node at (3.3+11,-1.9) [] {$\Gamma$};

\end{tikzpicture}

\caption{Performing the transformation $\b{x} \mapsto \b{x}'$, a rotation around $x-r\bhat v$, the QBX expansion center, such that $\bhat v'=\hat x_1' \b e_1$.
Without loss of generality, $\b x-r\bhat v$ is chosen
as the origin in the left-hand sub-figure.
}
\label{fig:rotational-sym1}
\end{figure}

Using the assumption that $G$ is rotationally symmetric, the cross-terms and their associated computational expense can be avoided. We perform a coordinate rotation around the center $\b{x}-r \b{\nu}$ such that $\b{\nu} \mapsto \b{\nu'} = \hat x_1'$ (see Figure~\ref{fig:rotational-sym1}). Under this transformation, $\b{x} \mapsto \b{x}', \b{y} \mapsto \b{y}'$, where $\b{x}', \b{y}'$ refer to the coordinates of $\b{x}, \b{y}$ in the rotated frame. Performing this rotation implies that \eqref{eq:bad-chain} is equivalent to
\begin{equation}
\label{eq:good-chain}
\left. \left( \frac{d^n}{dt^n} K(\b{x}' + t\hat x_1, \b{y}')  \right)\right|_{t=-r}.
\end{equation}
If $z_1 \in \R$, then \eqref{eq:good-chain} is equivalent to $\frac{d^n}{dt^n} \overline{f}(\overline{g}(t))$ where:
\[
\overline{f} : z_1 \mapsto K((z_1,x_2, \dots, x_d),\b{y})
\]
and
\[
\overline{g}: t \mapsto x_1 +t.
\]
Thus $\overline{f}: \R \to \R$ and $\overline{g}: \R \to \R$, and computing $\frac{d^n}{dt^n} \overline{f}(\overline{g}(t))$ does not require a multidimensional Faa Di Bruno rule or any sort of chain rule since derivatives of $\overline{g}$ are trivial. The line expansion for a source ($\b{y'}$)-target ($\b{x'}$) pair after rotation will look like:
\begin{equation}
\label{eq:flops}
\sum_{i=0}^{p_{\text{QBX}}} \left. \left( \frac{d^i}{dt^i} K(\b{x}' + t\hat x_1, \b{y}')  \right)\right|_{t=-r} \frac{r^i}{i!}.
\end{equation}

In practice, to use the line expansion formulation of QBX given by \eqref{eq:qbx}, we will need to perform a rotation for each source-target pair as described above to avoid the unnecessary computation of multi-dimensional chain rule cross-terms. In the QBX literature, approaches like the above have come to be called ``target-specific'' expansions. In effect, the approach described above represents a generalization of this concept to the setting of Cartesian/Taylor expansions.

Target-specific QBX expansions in the literature \cite{Siegel_2018,wala2020optimization} have been based on spherical harmonics and only indirectly perform these rotations before evaluating derivatives of the kernel for source-target pairs. This is because their expression for the layer potential contribution for a source-target pair is coordinate-independent and depends on the angle between the center-target and center-source vectors. The identity
\[
\sum_{n=0}^p \sum_{m=-n}^n L^m_n|\b{t}-\b{c}|^n Y^{m}_n(\theta_{\b{t}-\b{c}}, \phi_{\b{t}-\b{c}}) = \frac{1}{4\pi} \sum_{n=0}^p \frac{|\b{t}-\b{c}|^n}{|\b{s}-\b{c}|^{n+1}} P_n(\cos \gamma)
\]
((18) in \cite{wala2020optimization})
permits a computation of the truncated kernel expansion 
for a source at $\b{s} \in \R^3$, target at $\b{t} \in \R^3$, and center at $\b{c} \in \R^3$ while only considering $\gamma$, the angle between the center-target and center-source vectors. While the Cartesian formulation does not make the dependency on angle as explicit (and thus is most conveniently applied via the actual rotation), the benefit in computational cost is asymptotically identical and comparable in practice.

\subsection{Numerical Experiments on Rotation/Recurrence-Augmented QBX}

\begin{figure}[h]
\centering
\includegraphics{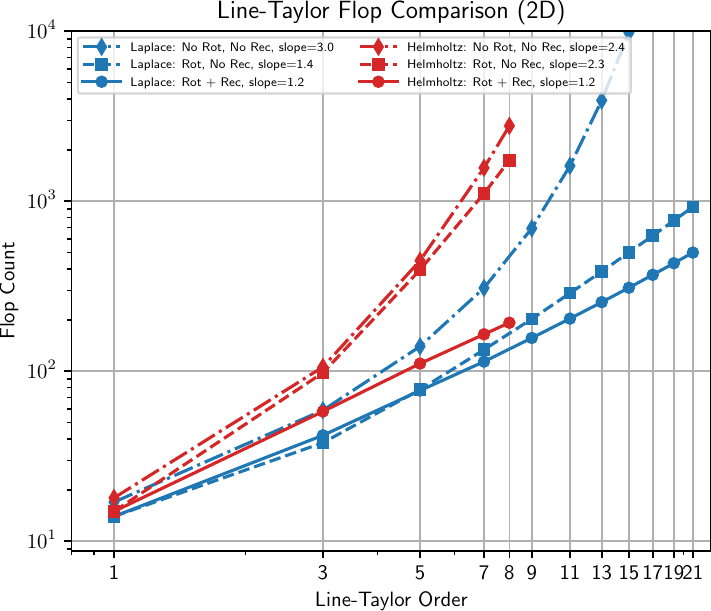}

\caption{Floating-point cost for computing a line-Taylor expansion \eqref{eq:flops} of varying order for Laplace and Helmholtz in 2D. Three configurations are compared: (1)~no rotation and no recurrence (dash-dot), (2)~rotation but no recurrence (dashed), and (3)~rotation with recurrence (solid). Log-log best-fit slopes are shown in the legend. Common subexpression elimination was used in all cases to reduce the number of flops.}
\label{fig:cost-2d}
\end{figure}

\begin{figure}[h]
\centering
\includegraphics{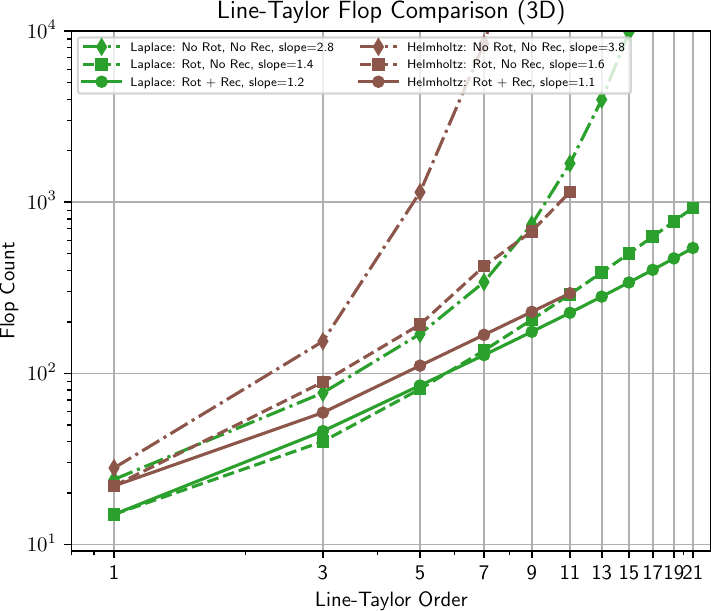}

\caption{Floating-point cost for computing a line-Taylor expansion \eqref{eq:flops} of varying order for Laplace and Helmholtz in 3D. Configurations and legend as in Figure~\ref{fig:cost-2d}.}
\label{fig:cost-3d}
\end{figure}

\subsubsection{Computational Cost}

For a given source-target pair, we consider the computational cost of computing a line expansion in \eqref{eq:flops} under three configurations: (1) without rotation and without recurrence, (2) with rotation but without recurrence, and (3) with rotation and with recurrence. When a recurrence is \emph{not} in use, the expansion is formed by computing the derivatives \eqref{eq:bad-chain} or \eqref{eq:good-chain} via Sympy~\cite{10.7717/peerj-cs.103}. In all cases, common subexpression elimination (CSE, also via Sympy) is used to reduce the number of floating point operations. For an order-$p_{\text{QBX}}$ line expansion, we must compute $O(p_{\text{QBX}})$ derivatives. When using a recurrence, we can expect each derivative to have amortized $O(1)$ cost, and thus $O(p_{\text{QBX}})$ complexity to compute an order-$p_{\text{QBX}}$ line expansion. Without recurrences, a $p$th order derivative can cost $O(p)$ or more, leading to a $O(p_{\text{QBX}}^2)$ or higher cost to compute a line expansion of order $p_{\text{QBX}}$.

In Figures~\ref{fig:cost-2d} and~\ref{fig:cost-3d}, we show the number of floating point operations for Laplace and Helmholtz in both 2D and 3D. We observe that the rotation alone provides a significant cost reduction by eliminating multi-dimensional chain rule cross-terms. On top of this, recurrences provide a further reduction: the recurrence and no-recurrence approaches diverge immediately with respect to the line expansion order. The log-log best-fit slopes confirm that the rotation with recurrence approach achieves approximately linear cost in the expansion order $p_{\text{QBX}}$ across all four kernels tested (slopes $\approx 1.1$--$1.2$), consistent with the $O(p_{\text{QBX}})$ complexity predicted by the amortized $O(1)$ cost per recurrence step. By contrast, approaches without recurrences exhibit super-algebraic growth, with apparently-increasing slopes.

\subsubsection{Error}

When incorporating the recurrences developed here into QBX for layer potential evaluation, it is relevant to ask whether the additional error contribution (analyzed in Section~\ref{sec:hybrid-error}) results in a noticeable error increase in the overall scheme. To investigate this, we compute the relative $L^\infty$ error
\begin{equation}
\label{eq:qbx-error}
E(u) = \frac{\|u - u_{\text{true}}\|_\infty}{\|u_{\text{true}}\|_\infty}
\end{equation}
in single layer potential evaluation on the ellipse $(2\cos t, \sin t)$ for $t \in [0, 2\pi]$ (aspect ratio 2:1) for the Laplace PDE in 2D. The ellipse is discretized using composite order-16 Gauss-Legendre quadrature with $n_p \in \{60, 200, 360\}$ panels. The QBX expansion centers are placed along the inward normal at a distance $2.5\,h$ from the boundary, where $h = 2\pi/n_p$ is the panel size. The true single layer potential $u_{\text{true}}$ is obtained from the known eigenvalue $\mu_n = \frac{1}{2n}(1 + ((1-r)/(1+r))^n)$ of the single layer operator on the ellipse with eccentricity $r = 1/a$ and oscillating source density $\rho(t) = \cos(nt)$ with $n=10$. We compare QBX only ($u_{\text{qbx}}$) versus incorporating recurrences to compute the derivatives in \eqref{eq:qbx} ($u_{\text{qbxrec}}$) for different mesh resolutions $h$ and QBX truncation orders $p_{\text{QBX}}$. For a description of $m$ and $p_{\text{small-$|x_1|$}}$ see Sections~\ref{sec:small-x1} and \ref{sec:error-modeling}.

In Figure~\ref{figure:QBXLaplace}, we find that for $p_{\text{QBX}} \in \{5, 7\}$, our recurrences do not add appreciable error on top of the existing solution error compared to using QBX to compute layer potentials. For $p_{\text{QBX}} \in \{9, 11\}$, the recurrence-based evaluation floors at approximately one digit above machine precision, losing about one digit compared to standard QBX at the finest mesh resolutions.

\begin{figure}[h]
\centering
\includegraphics{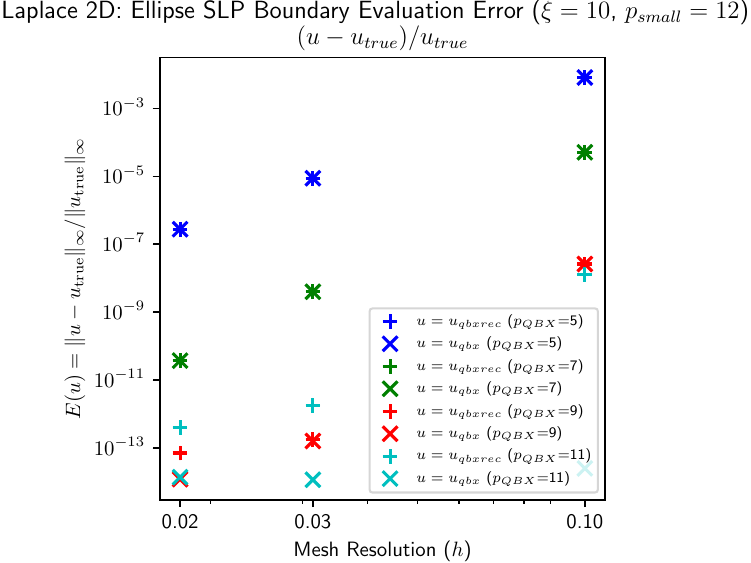}
\caption{The relative error $E(u)$ for evaluating the single layer potential for the Laplace PDE in 2D on the boundary of an ellipse $(2\cos(t), \sin(t))$, $t \in [0, 2\pi]$, with oscillating source density $\rho(t)=\cos(10t)$, incorporating recurrences ($u_{\text{qbxrec}}$) versus using QBX only ($u_{\text{qbx}}$) for different panel sizes $h$ and QBX orders. Here $\xi$ is the dispatch boundary between the large-$|x_1|$ recurrence and the small-$|x_1|$ expansion (see Section~\ref{sec:hybrid-error}), and $p_{\text{small-$|x_1|$}}$ is the truncation order of the small-$|x_1|$ expansion.}
\label{figure:QBXLaplace}
\end{figure}
\section{Conclusion}
\label{sec:conclusions}

We have presented an efficient approach for evaluating high-order derivatives of Green's functions using a hybrid numerical-symbolic algorithm. We provide error analysis for our overall algorithm for derivative evaluation. We support our error analysis with empirical results for the Helmholtz, Laplace, and Biharmonic PDE in 2D, with the error behavior in 3D expected to behave analogously.

We have further introduced a new rotation-based method for target-specific QBX evaluation in the Cartesian setting that attains dramatically lower cost than existing symbolic approaches that composes cleanly with the recurrence scheme. Through the combination of both schemes, a significant reduction in the computational cost of QBX is attained. Finally, we find that these improvements, when incorporated into QBX, contribute little to no error.

Extensions of this work of immediate interest include the availability of recurrences for Green's functions with relaxed symmetry assumptions as well as the numerical properties of such recurrences.

\appendix
\section{Numerical Experiments Supporting Assumptions \ref{ass:deriv-growth-odd}--\ref{ass:deriv-growth-cubed-odd-even}}
\label{app:assumption-numerics}

We design the following numerical experiment to support Assumptions \ref{ass:deriv-growth-odd}--\ref{ass:deriv-growth-cubed-odd-even}. We plot a heat map (with $x_1$ and $x_2$ coordinates shown logarithmically) of the expression for Laplace/Helmholtz PDE in 2D:
\begin{equation}
\label{eq:app-laplace-odd}
\left|\max_{0\leq \xi\leq x_1}\left(\partial^{\mu}_{x_1}G(|\b{x}|)\right) \right| / \left(\frac{|x_1|}{\overline{x}^{\mu+1}} \right) \qquad\text{ (Laplace/Helmholtz 2D)}
\end{equation}
when $\mu = 5$ and the expression
\begin{equation}
\label{eq:app-laplace-even}
\left|\max_{0\leq \xi\leq x_1}\left(\partial^{\nu}_{x_1}G(|\b{x}|)\right) \right| / \left(\frac{1}{\overline{x}^{\nu}} \right) \qquad\text{ (Laplace/Helmholtz 2D)}
\end{equation}
when $\nu = 6$. We plot a heat map of the expression for Biharmonic PDE in 2D:
\begin{equation}
\label{eq:app-biharm-odd}
\left|\max_{0\leq \xi\leq x_1}\left(\partial^{\mu}_{x_1}G(|\b{x}|)\right) \right| / \left(\frac{|x_1^3|}{\overline{x}^{\mu+1}} \right) \qquad\text{ (Biharmonic 2D)}
\end{equation}
when $\mu = 5$ and the expression
\begin{equation}
\label{eq:app-biharm-even}
\left|\max_{0\leq \xi\leq x_1}\left(\partial^{\nu}_{x_1}G(|\b{x}|)\right) \right| / \left(\frac{|x_1^2|}{\overline{x}^{\nu}} \right) \qquad\text{ (Biharmonic 2D)}
\end{equation}
when $\nu = 6$.

From these, we observe the following statements regarding the claimed bounds:
\begin{itemize}
\item Figure~\ref{fig:fig10} suggests that for Laplace/Helmholtz 2D for $\mu = 5$ when $|x_1|/\overline{x} < 1$:
\[
1\leq \left|\max_{0\leq \xi\leq x_1}\left(\partial^{\mu}_{x_1}G(|\b{x}|)\right) \right| / \left(\frac{|x_1|}{\overline{x}^{\mu+1}} \right) \leq 100 \qquad\text{ (Laplace/Helmholtz 2D)},
\]
and
\item Figure~\ref{fig:fig10} also suggests for Biharmonic 2D for $\mu = 5$ when $|x_1|/\overline{x} < 1$:
\[
1\leq \left|\max_{0\leq \xi\leq x_1}\left(\partial^{\mu}_{x_1}G(|\b{x}|)\right) \right| / \left(\frac{|x_1^3|}{\overline{x}^{\mu+1}} \right) \leq 100. \qquad\text{ (Biharmonic 2D)}
\]
\item Figure~\ref{fig:fig11} suggests that for Laplace/Helmholtz 2D for $\nu = 6$ when $|x_1|/\overline{x} < 1$:
\[
1\leq \left|\max_{0\leq \xi\leq x_1}\left(\partial^{\nu}_{x_1}G(|\b{x}|)\right) \right| / \left(\frac{1}{\overline{x}^{\nu}} \right) \leq 100 \qquad\text{ (Laplace/Helmholtz 2D)},
\]
and
\item Figure~\ref{fig:fig11} also suggests for Biharmonic 2D for $\nu = 6$ when $|x_1|/\overline{x} < 1$:
\[
1\leq \left|\max_{0\leq \xi\leq x_1}\left(\partial^{\nu}_{x_1}G(|\b{x}|)\right) \right| / \left(\frac{|x_1|^2}{\overline{x}^{\nu}} \right) \leq 100 \qquad\text{ (Biharmonic 2D)}
\]
\end{itemize}

\begin{figure}[h]
\centering
\includegraphics{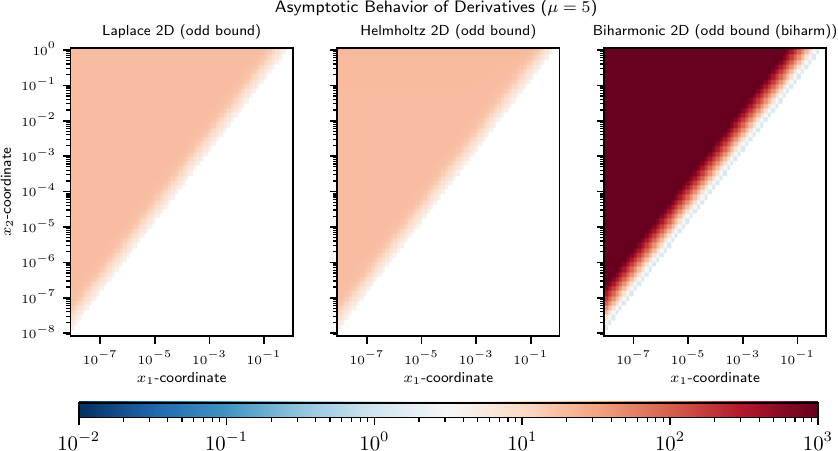}

\caption{Numerical validation of Assumptions~\ref{ass:deriv-growth-odd} and~\ref{ass:deriv-growth-cubed-odd} for odd derivative order $\mu=5$. The ratio \eqref{eq:app-laplace-odd} (Laplace/Helmholtz) and \eqref{eq:app-biharm-odd} (Biharmonic) is shown on a log-log grid. The region $|x_1|/\overline{x} \geq 1$ is blanked (white) because the assumptions and the small-$|x_1|$ error bounds they support only apply when $|x_1|/\overline{x} < 1$. In the displayed region, the ratio remains bounded, supporting the assumed asymptotic growth rates.}
\label{fig:fig10}
\end{figure}

\begin{figure}[h]
\centering
\includegraphics{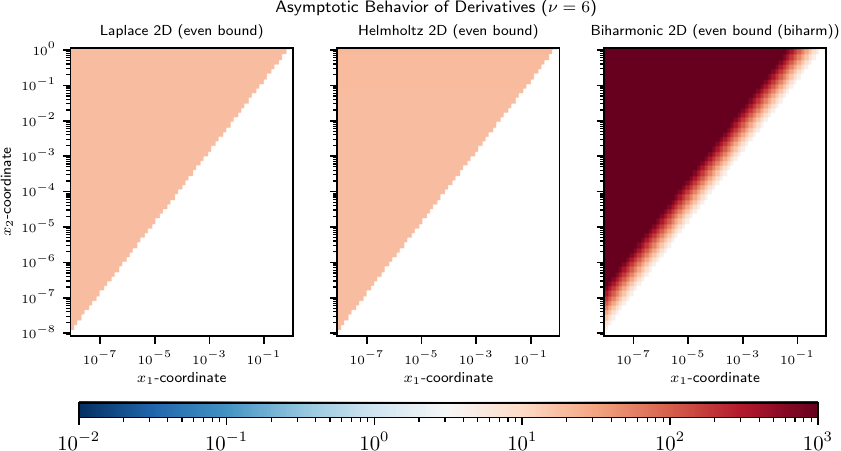}

\caption{Numerical validation of Assumptions~\ref{ass:deriv-growth-odd-even} and~\ref{ass:deriv-growth-cubed-odd-even} for even derivative order $\nu=6$. The ratio \eqref{eq:app-laplace-even} (Laplace/Helmholtz) and \eqref{eq:app-biharm-even} (Biharmonic) is shown on a log-log grid. The region $|x_1|/\overline{x} \geq 1$ is blanked (white) because the assumptions and the small-$|x_1|$ error bounds they support only apply when $|x_1|/\overline{x} < 1$. In the displayed region, the ratio remains bounded, supporting the assumed asymptotic growth rates.}
\label{fig:fig11}
\end{figure}

\section*{Acknowledgments}
The authors' research was supported by the National Science Foundation
under grants SHF-1911019 and DMS-2410943, by the US Department of Energy
under contract DE-NA0003963, as well as by the Siebel School of Computing
and Data Sciences at the University of Illinois at Urbana-Champaign.
The authors would like to thank Shawn Lin and Xiaoyu Wei for helpful discussions.

\bibliographystyle{plainurl} 
\bibliography{refs} 

\end{document}